\documentclass{article} 

\usepackage{parskip}
\usepackage{float}      
\usepackage{graphicx}   
\usepackage{hyperref}   
\usepackage{latexsym}   
\usepackage{xcolor}     
\usepackage{lipsum}     
\usepackage{amsfonts}   
\usepackage{amsthm}

\usepackage{amssymb}    
\usepackage{amsmath}    
\usepackage{bbm}        
\usepackage{eucal}      
\usepackage{mathabx}    
\usepackage{mathrsfs}   
\usepackage{mathtools}

\usepackage{multirow}   
\usepackage{array}      
\usepackage{booktabs}   
\usepackage{ulem}       

\usepackage[numbers,sort&compress]{natbib}  

\usepackage{todonotes}  

\usepackage{enumerate}  
\usepackage{enumitem}   

\newtheorem{theorem}{Theorem}[section]
\newtheorem{definition}[theorem]{Definition}

\newtheorem{proposition}[theorem]{Proposition}
\newtheorem{lemma}[theorem]{Lemma}

\newtheorem{assumption}{Assumption}

\newtheorem{prop}[theorem]{Proposition}

\newtheorem{rem}[theorem]{Remark}

\renewcommand{\theassumption}{\Alph{assumption}}

\makeatletter
\newcommand{\settheoremtag}[1]{
\let\oldtheassumption\theassumption
\renewcommand{\theassumption}{#1}
\g@addto@macro\endassumption{
\global\let\theassumption\oldtheassumption}
}
\makeatother

\def\wt{\widetilde}
\def\wh{\widehat}
\def\wb{\overline}

\newcommand{\norm}[1]{{\vert\kern-0.25ex\vert #1 \vert\kern-0.25ex\vert}}
\newcommand{\bignorm}[1]{{\big\vert\kern-0.25ex\big\vert #1 \big\vert\kern-0.25ex\big\vert}}
\newcommand{\opnorm}[1]{{\vert\kern-0.25ex\vert\kern-0.25ex\vert #1 \vert\kern-0.25ex\vert\kern-0.25ex\vert}}

\newcommand{\dd}{\mathrm{d}}

\newcommand{\DeltaH}{\Delta_h}
\newcommand{\PX}{\mathbb{P}}
\newcommand{\E}{\mathbb{E}}

\newcommand*{\transp}[1]{%
\ensuremath{%
\mskip1mu
\prescript{\smash{\mathrm t}}{}{%
    \mathstrut#1%
}%
}%
}

\newcommand{\bflambda}{{\boldsymbol{\lambda}}}
\newcommand{\bfLambda}{{\boldsymbol{\Lambda}}}
\newcommand{\bfmu}{{\boldsymbol{\mu}}}
\newcommand{\bfphi}{{\boldsymbol{\phi}}}

\usepackage{fancyhdr}

\title
{A unified theory of order flow, market impact, and volatility\footnote{Youssef Ouazzani Chahdi, Mathieu Rosenbaum and Gr\'egoire Szymanski gratefully acknowledge support from the ILB Chair \textit{Artificial Intelligence and Quantitative Methods for Finance} at University Paris Dauphine-PSL. The authors thank Pavel Chigansky and Marina Kleptsyna for key input on mixed fractional Brownian motion. They are also grateful to Jean-Philippe Bouchaud and Kevin Webster for inspiring discussions on order flow modeling, and thank BMLL Technologies for providing the historical market data used in this study.}}

\author
{Johannes Muhle-Karbe\footnote{Department of Mathematics, Imperial College London,
\texttt{j.muhle-karbe@imperial.ac.uk
}} \and Youssef Ouazzani Chahdi\footnote{MICS, CentraleSupélec, \texttt{youssef.ouazzani-chahdi@centralesupelec.fr}} \and Mathieu Rosenbaum\footnote{Ceremade, Universit\'e Paris Dauphine-PSL, \texttt{mathieu.rosenbaum@dauphine.psl.eu}} \and Gr\'egoire Szymanski\footnote{DMATH, Université du Luxembourg, \texttt{gregoire.szymanski@uni.lu}}}

\allowdisplaybreaks

\begin{document}

\date{\today}

\maketitle


\renewcommand{\thefootnote}
{\arabic{footnote}} 

\begin{abstract}
We propose a microstructural model for the order flow in financial markets that distinguishes between {\it core orders} and {\it reaction flow}, both modeled as Hawkes processes. This model has a natural scaling limit that reconciles a number of salient empirical properties: persistent signed order flow, rough trading volume and volatility, and power-law market impact. In our framework, all these quantities are pinned down by a single statistic $H_0$, which measures the persistence of the core flow. Specifically, the signed flow converges to the sum of a fractional process with Hurst index $H_0$ and a martingale, while the limiting traded volume is a rough process with Hurst index $H_0-1/2$. No-arbitrage constraints imply that volatility is rough, with Hurst parameter $2H_0-3/2$, and that the price impact of trades follows a power law with exponent $2-2H_0$. 
The analysis of signed order flow data yields an estimate $H_0 \approx 3/4$. This is not only consistent with the square-root law of market impact, but also turns out to match estimates for the roughness of traded volumes and volatilities remarkably well.

\end{abstract}



\newpage 

\section{Introduction}

Prices and traded quantities are the fundamental observables in any financial market. In a vast body of research initiated by Bachelier~\cite{bachelier1900theorie} and Black and Scholes~\cite{black1973pricing}, (semi-)martingales have emerged as the canonical model for asset prices, reflecting the absence of arbitrage and limited predictability of returns. In contrast, there is no similar standard model class for the corresponding order flow yet. A key challenge is that any such model must at the same time capture the stylized properties of traded amounts (``unsigned volumes'') and their directionality (``signed order flow''). Moreover, through the price impact of trades, order flow and price dynamics are intimately linked. A consistent model for the order flow therefore must strike a delicate balance to consistently connect the salient features of several distinct datasets. The present study sets out to do this in a principled yet parsimonious manner.

To explain more clearly what is the challenge at hand, let us first briefly review some of the statistical regularities of order flow data that are very robust across different markets, assets, and time periods: 
\begin{itemize}
\item[$\bullet$] 
\textbf{Persistent order flow.}  
The signed order flow exhibits significant persistence and long memory, commonly attributed to order splitting, sustained trading programs, and long-lived trading motives~\cite{lo.wang.00,lillo2004long,bouchaud2008marketsslowlydigestchange,LilloMikeFarmer2005,farmer2006market,sato2023inferring}. Its sample paths also display a markedly smoother behavior than Brownian motion on longer timescales, cf.~Figure~\ref{fig:signed_LVMH}.
\item[$\bullet$] \textbf{Rough traded volume.}  
In contrast, the unsigned traded amounts have a much rougher temporal structure, cf.~Figure~\ref{fig:unsigned_LVMH}, in line with the ``rough volatility models'' introduced by~\cite{gatheral2018volatility}. The similarity between trading activity and integrated variance is well established~\cite{karpoff.87,madhavan1997security,wyart2008relation,dayri2015large,kyle2016market}, but raises the question how signed flow and unsigned volumes can be modeled consistently. Indeed, persistent order flow naturally suggests models driven by fractional Brownian motions with Hurst indices way above $1/2$, reflecting long-range dependence. In contrast, rough volatility and unsigned volume are best captured by fractional processes with Hurst indices far below $1/2$. 
\end{itemize}

\begin{figure}[H]
    \centering
    \includegraphics[width=0.9\linewidth]{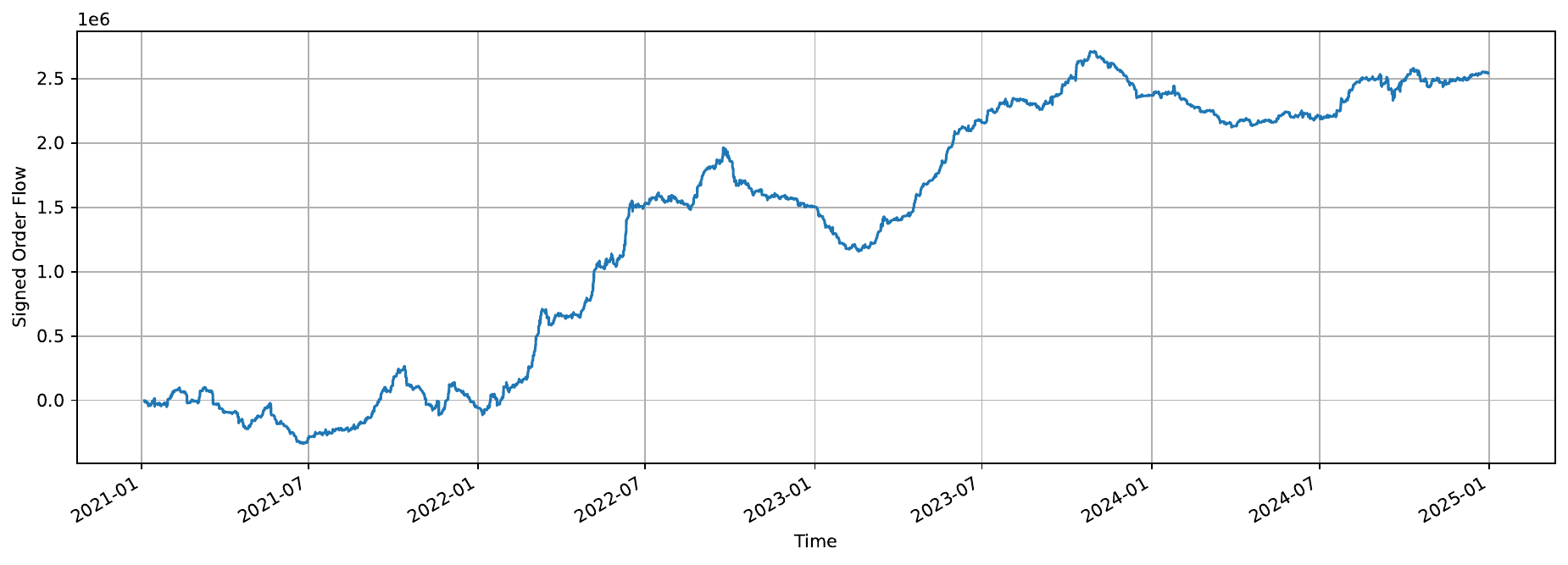}
    \caption{Cumulative signed order flow of the representative stock LVMH between 2021 and 2024.}
    \label{fig:signed_LVMH}
\end{figure}

\begin{figure}[H]
    \centering
    \includegraphics[width=0.9\linewidth]{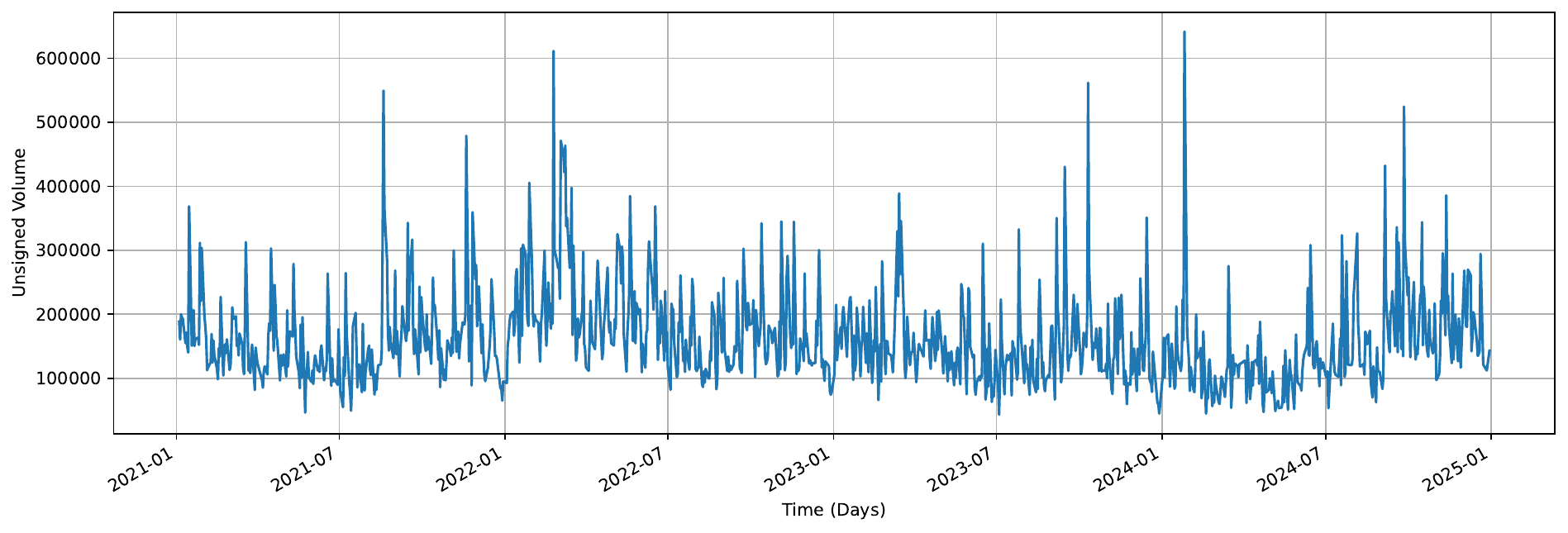}
    \caption{Daily traded volume of the representative stock LVMH between 2021 and 2024.} 
    \label{fig:unsigned_LVMH}
\end{figure}

This modeling challenge is exacerbated by the multiscale behavior of the signed order flow. This is illustrated in Figure~\ref{fig:single_fbm}, which reports Hurst parameter estimates for the signed order flow sampled at different frequencies. At high frequencies, the estimates are close to $0.5$, consistent with diffusive models \cite{guasoni.weber.17,carmona.webster.19}. 
As the sampling frequency decreases, estimated Hurst exponents increase steadily and reach values around $0.65$ when sampling hourly, in line with smooth fractional models. This suggests that no ``pure'' fractional Brownian motion model adequately captures the dynamics of the order flow. 

\begin{figure}[H]
    \centering
    \includegraphics[width=0.8\linewidth]{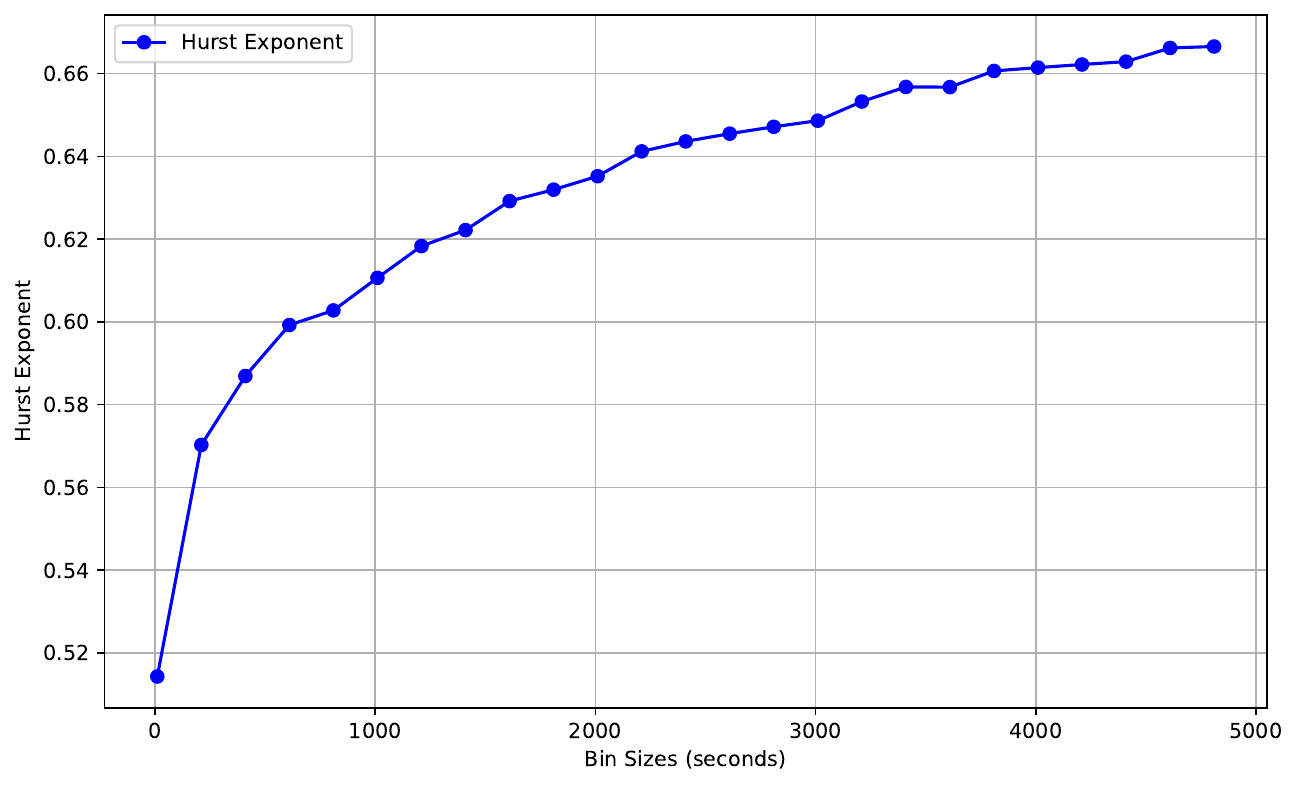}
    \caption{Average Hurst exponent estimates for signed order flow over 40 stocks for the period 2021--2024, under fractional Brownian motion specification.\\
     \footnotesize{\it Note}: The data used for the estimations throughout the paper are described in Appendix \ref{app:data}.}
    \label{fig:single_fbm}
\end{figure}

Moreover, as already alluded to above, a consistent model for the order flow must not only recapture its own empirical properties but, via the price impact of trades, also remain consistent with price dynamics. More specifically, when prices are (close to) martingales and permanent price impact is linear to preclude statistical arbitrage~\cite{huberman2004price,gatheral2010no,toth2011anomalous,farmer2013efficiency,donier2015fully,benzaquen2018market}, then the expected future order flow also pins down prices and their volatility, as well as the price impact decay kernel that describes how the impact of each trade dissipates over time~\cite{jusselin2020noarbitrage}. 

Any consistent model for the order flow therefore must be consistent with the large bodies of evidence accumulated for the universal scaling properties of price volatility and price impact:
\begin{itemize}
\item[$\bullet$] 
\textbf{Rough volatility.} As indicated above, volatility time series exhibit very rough sample paths. Across a very wide range of markets, this leads to estimates for Hurst parameters in a range between $0.05$ and $0.15$~\cite{gatheral2018volatility,chong2022minimax,bolko2023gmm,chong2025nonparametric, shi2025fractional, han2025rate}.\footnote{The rough volatility paradigm sheds new light on important volatility features such as long memory \cite{lo1991long, comte1996long, bennedsen2022decoupling, li2025weak}, and motivates a new generation of stochastic volatility models. 
These developments have profoundly reshaped both theoretical modeling and practical applications, with notable implications for derivatives pricing \cite{bayer2016pricing} 
and volatility forecasting \cite{gatheral2018volatility,bennedsen2022decoupling,wang2024optimal}.}
\item[$\bullet$] 
\textbf{Universal market impact scaling.}  
The average price response to a large order follows a remarkably stable scaling law: impact grows approximately as the square root of the traded size \cite{loeb.83,BouchaudEtAl2004Fluctuations,toth2011anomalous,bershova2013non,frazzini.al.18, Sato_2025}. This empirical regularity again appears largely invariant across markets and trading regimes.  
\end{itemize}




In the present study, we build a unifying model that reconciles all of these empirical facts in a consistent yet extremely parsimonious manner. Our analysis is based on a two-layer model for the order flow that distinguishes between {\it core flow} and  {\it reaction flow}. Core flow captures autonomous trading activity arising from slow-moving investment decisions, portfolio rebalancing, or fundamental information.\footnote{Note that the idea of a component of order flow that is insensitive to contemporaneous trading activity is also at the heart of \cite{LilloMikeFarmer2005}, a paper whose conclusions have recently been empirically confirmed in \cite{sato2023inferring}.} The reaction flow, by contrast, represents responses to observed market activity, including liquidity provision, market making, and high-frequency trading.  Both layers are modeled using Hawkes processes, following a now widely-used and empirically well-supported approach in market microstructure \cite{jaisson2015limit,jaisson2016rough,horst2022microstructure}. This modeling choice is not motivated by behavioral assumptions but primarily by statistical adequacy: Hawkes dynamics reproduce with remarkable accuracy the clustering, persistence and scaling properties of order arrivals observed in markets.\footnote{A closely related line of work is developed in \cite{maitrier2025subtleinterplaysquarerootimpact,maitrier2025subtle2,naviglio2025estimation}, where detailed models of metaorders and generalized propagator frameworks are constructed to reconcile long memory in order flow with price diffusivity and square-root impact. Our approach is complementary: rather than modeling the fine structure of individual metaorders, we adopt a reduced-form, statistical perspective focused on aggregate trading flows and their scaling limits.}

While the resulting microstructural model is flexible enough to fit a wide range of empirical features, it a priori involves many degrees of freedom. Our key insight is that imposing the existence of a non-degenerate scaling limit drastically restricts the admissible parameter space. This leads to an extremely parsimonious limiting model.

More specifically, our first main result establishes that in a large time asymptotic, the suitably rescaled signed order flow converges to the sum of two terms: a fractional component with Hurst index $H_0>1/2$, inherited from the persistence of the core flow, and a martingale component generated by the reaction flow. This limiting process is close to a ``mixed fractional Brownian motion'', that is the sum of a standard Brownian motion and an independent fractional Brownian motion of~\cite{cheridito2001mixed}. Deployed to model order flow, this mixed fractional structure naturally explains why Hurst exponent estimates for signed order flow are strongly scale dependent: they are close to $0.5$ at high frequencies, where the memoryless component dominates, and increase substantially at coarser time scales, where the persistent fractional component becomes more and more visible. 
As a consequence, classical roughness estimators based on a pure fractional Brownian motion specification are inherently biased in this setting, in line with the results displayed in Figure~\ref{fig:single_fbm}. By contrast, when inference procedures explicitly account for the mixed fractional structure, Hurst parameter estimates become remarkably stable across aggregation scales. Figure~\ref{fig:mixed_fbm_intro} illustrates this phenomenon: for all considered sampling frequencies, the estimated Hurst exponent $H_0$ under mixed fractional specification lies in the narrow range $0.75$--$0.80$.

\begin{figure}
    \centering
    \includegraphics[width=0.75\linewidth]{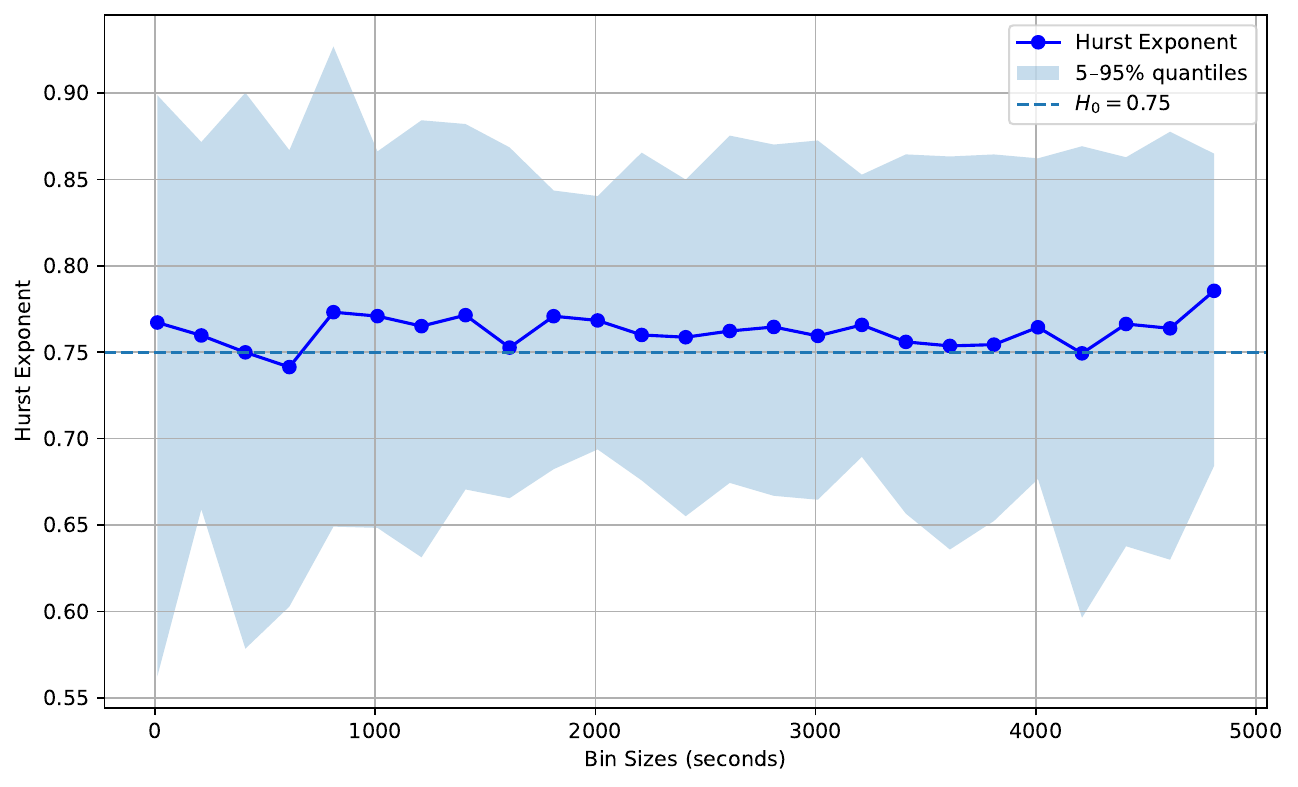}
    \caption{Average Hurst exponent estimates for signed order flow over $40$ stocks for the period 2021--2024, under mixed fractional Brownian motion specification.\\
    \footnotesize{\it Note}:  For each asset and time scale $\Delta$, the Hurst parameter of the fractional Brownian motion component is estimated using quadratic variations computed at the time scales $\Delta$, $2\Delta$, and $4\Delta$. The estimation procedure then follows the methodology developed for mixed fractional processes in \cite{CDM22,szymanski2026mixed}.}
    \label{fig:mixed_fbm_intro}
\end{figure}

Our second main result provides the corresponding asymptotic behavior for the unsigned traded volume. 
It turns out to be primarily driven by the reaction flow. Yet, the requirement of a non-trivial scaling limit imposes tight constraints on the latter, so that the memory parameter $H_0$ governing the persistence of the core flow also directly determines the statistical nature of the endogenous reaction intensity. More specifically, the (cumulative) unsigned volume converges to (the integral of) a rough process with Hurst exponent $H_0 - 1/2$. 
This result provides a structural explanation for the empirically observed roughness of traded volumes. On our dataset, standard autocovariance-based estimators yield Hurst exponent estimates for unsigned volume in the range $0.15$--$0.35$, see Figure~\ref{fig:unsigned_volume_hurst}. 
These empirical values match our theoretical predictions when $H_0$ is of order $0.75$, which coincides remarkably well with the estimates obtained for signed order flow.

Taken together, these results provide a consistent explanation for the observed scaling properties of signed order flow and unsigned volume. But beyond this, the mixed fractional structure of the order flow also plays a crucial role in explaining the joint dynamics of order flow, market impact, and volatility. To wit, in the relevant regime $H_0 > 3/4$, the mixed fractional Brownian motion admits a semi-martingale representation \cite{cheridito2001mixed}. 
This property allows us to connect order flow dynamics to price formation while preserving the martingale property of prices. Exploiting fine regularity properties of the drift component in this representation, we show that the same parameter $H_0$ governing the persistence of the core flow also controls the scaling behavior of both market impact and volatility.

\begin{figure}[htbp]
    \centering
\includegraphics[width=0.75\linewidth]{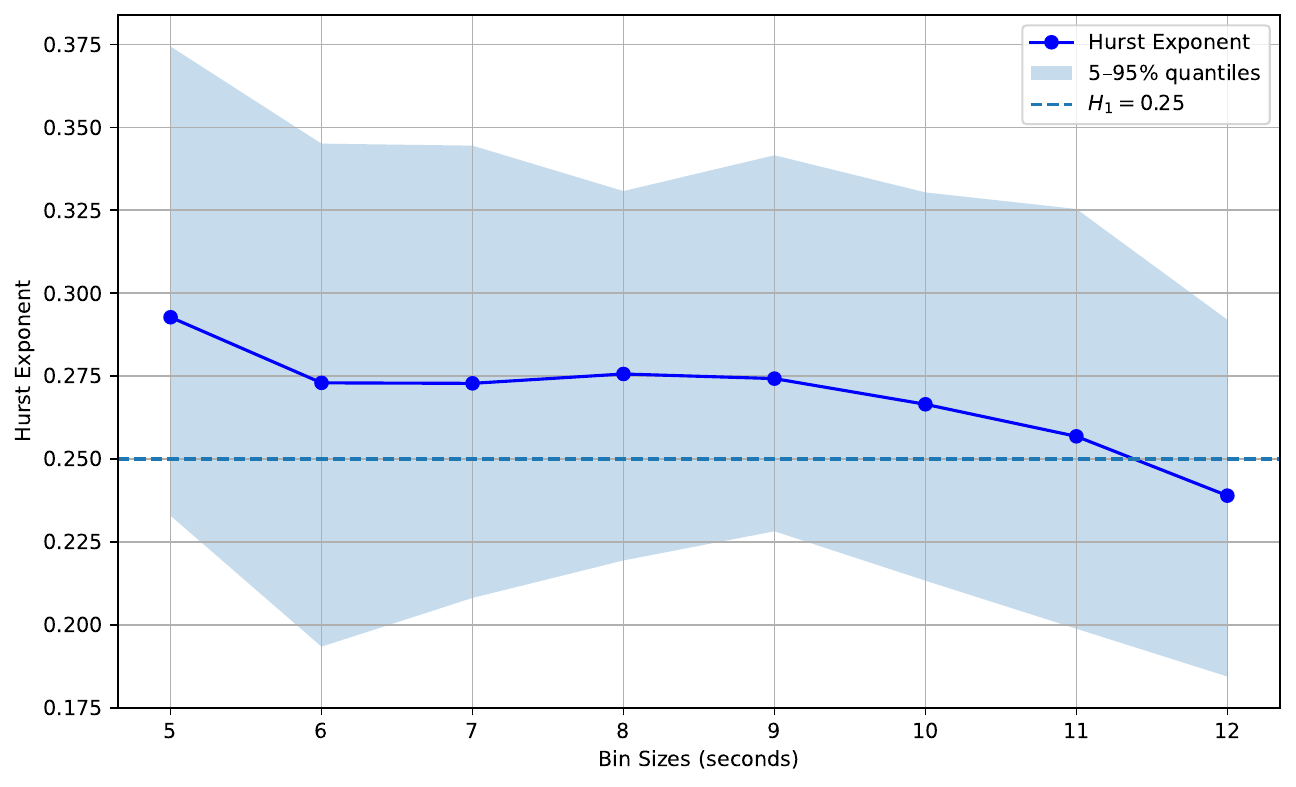}
    \caption{
Average Hurst exponent estimates for unsigned trading volume, averaged over $40$ stocks for the period 2021--2024.\\
\footnotesize{\it Note}: For each asset and time scale $\Delta$, the procedure is as follows: the total traded volume is aggregated over bins of size $\Delta$; the intraday seasonal pattern is removed multiplicatively; volume increments are computed and truncated at three times their standard deviation to mitigate the impact of outliers and exclude potential jumps in the volume intensity process; the auto-covariance function is then estimated, and the Hurst exponent is obtained using a GMM-based approach similar to \cite{li2016generalized}. The methodology closely follows the procedures developed for volatility analysis in \cite{chong2022CLT,chong2026intraday}. 
}
    \label{fig:unsigned_volume_hurst}
\end{figure}

More precisely, following the approach in \cite{jusselin2020noarbitrage} where prices reflect the anticipation of future order flow, no-arbitrage arguments imply that the impact function follows a power law with exponent $2-2H_0$. This shape of impact then implies that volatility behaves as a rough fractional process with Hurst exponent $2H_0-3/2$.

Thus, persistence in signed order flow, roughness of unsigned volume, roughness of volatility, and  market impact exponent are not independent empirical features but are jointly determined by a single structural parameter $H_0$. In the empirically relevant regime $H_0 \approx 3/4$, our framework therefore recovers, in a unified and internally consistent setting, strongly persistent order flow, rough unsigned volume with Hurst exponent near $1/4$, very rough volatility with Hurst parameter close to zero, and the square-root law of market impact. Rather than relying on finely tuned assumptions, this joint behavior emerges as a consequence of scale separation, no-arbitrage constraints, and the mixed fractional nature of aggregate order flow. 

The rest of the paper is organized as follows. 
Section~\ref{sec:model_description} introduces the two-layer Hawkes framework. Section~\ref{sec:scaling_limits} derives the scaling limits for the core, reaction, and aggregate order flows. Section~\ref{sec:links} establishes the connections between order flow, market impact, and volatility. 
For better readability, all proofs are collected in the appendix.

\section{A two-layer Hawkes model for order flow}
\label{sec:model_description}

This section introduces a model that allows to consistently capture the empirical properties of signed order flow and unsigned volume and to establish a unified framework connecting them to rough volatility and power-law market impact.

To this end, we decompose the aggregate order flow into two conceptually distinct building blocks:

\begin{itemize}

\item \textbf{Core order flow:} The {\it core order flow} arises from a heterogeneous mixture of of autonomous trading motives and horizons. In particular, it comprises medium and low frequency strategies, often grounded in fundamental information, long term valuation views, or trend following dynamics. Such strategies explain part of the empirically observed persistence in order flow, where the signs of trades exhibit long-range dependence. This is complemented by metaorder splitting: large institutional trades are executed incrementally over time to minimize market impact. 

\item \textbf{Reaction orders:} Unlike the core flow, {\it reaction orders} are not initiated for autonomous reasons but arise as a response to other trades. This applies both to the core flow (which contains both informed trades and trading opportunities) and to other reaction orders. Such reaction orders reflect the dynamic interplay among liquidity providers, high-frequency market makers and quantitative strategies that continuously adjust their positions and inventories. The resulting feedback mechanisms generate additional layers of dependence within the order flow, complementing the persistence directly induced by the core order flow.


\end{itemize}

\subsection{Core order flow}  
We model the core buy and sell orders by two independent univariate Hawkes processes, denoted by $F^+$ and $F^-$, respectively. This is a very natural modeling tool for the splitting of a metaorder or a trend following strategy. For example, once a child order is submitted, the probability of observing further orders of the same sign increases, reflecting the continuation of an execution program.  

Formally, both $F^+$ and $F^-$ have the same baseline intensity $\nu> 0$ and the same excitation kernel $\varphi_0 \colon \mathbb{R}_+ \to \mathbb{R}_+$ that governs the temporal dependence between trades. Hence, the intensities of core buy and sell orders are given by
\begin{equation*}
\lambda^{+}_t = \nu + \int_0^{t-} \varphi_0(t-s) \, dF^{+}_s,
\qquad
\lambda^{-}_t = \nu + \int_0^{t-} \varphi_0(t-s) \, dF^{-}_s.
\end{equation*}
When the excitation kernel $\varphi_0$ decays rapidly, the process approximates a memoryless sequence of orders, with limited interaction between successive trades. Conversely, a slowly decaying $\varphi_0$ implies that each trade continues to elevate the probability of subsequent trades in the same direction over an extended horizon. This captures that large metaorders or trend-following strategies, once initiated, generate persistent streams of transactions. 

The signed core order flow and unsigned core volume are in turn given by 
\begin{equation*}
    F_t = F_t^+ - F_t^-, \qquad V_t = F_t^+ + F_t^-,
\end{equation*}
which measure the directional flow and overall trading volumes due to core trading activity.

\subsection{Reaction orders}  

We now turn to the market's endogenous reaction to incoming orders. Because trading is anonymous, it is hard to discriminate autonomous core orders from other trades. Therefore, we model this reaction flow via Hawkes process driven by core and other reaction trades in the same manner. More specifically, we consider a two-dimensional Hawkes process
\begin{equation*}
{\mathbf{N}}_t = (N_t^+, N_t^-),
\end{equation*}
where $N^+$ describes reaction buys and $N^-$ models reaction sells.

The baseline intensity of $\mathbf{N}$ is driven by the reaction to core orders through a symmetric kernel matrix
\begin{equation*}
\bfphi = \begin{pmatrix}
    \varphi_1 & \varphi_2 \\
    \varphi_2 & \varphi_1
\end{pmatrix},
\end{equation*}
so that
\begin{equation*}
\bfmu_t = \int_0^t \bfphi(t-s) \cdot \, d\mathbf{F}_s
\qquad \text{where} \qquad
\mathbf{F}_t = (F_t^+, F_t^-).
\end{equation*}
The aggregate intensity of $\mathbf{N}$ is in turn given by
\begin{equation*}
\bflambda_t = \bfmu_t + \int_0^t \bfphi(t-s) \cdot \, d\mathbf{N}_s = \int_0^t \bfphi(t-s) \cdot \, \dd(\mathbf{F}_s + \mathbf{N}_s).
\end{equation*}
This structure of Hawkes process branching on Hawkes process describes the dynamics of the reaction flow:
\begin{itemize}
\item Following a core buy order at $t_0$, a wave of reaction buy orders (with intensity $\varphi_1$) is triggered on the ask side, reflecting for instance momentum strategies, while a wave of reaction sell orders (with intensity $\varphi_2$) may appear on the bid side, reflecting inventory rebalancing or contrarian liquidity provision;
\item The situation is symmetric for core sell orders, with $\varphi_1$ and $\varphi_2$ swapping roles.
\item Non-core orders are digested by the market through the same mechanism. This is represented by the integral term with respect to $\dd\mathbf{N}$ in the intensity of $\mathbf{N}$. The market uses exactly the same kernel to process core and non-core orders as there is no way to distinguish between them. 
\end{itemize}

\subsection{Aggregate order flow}  
The aggregate order flow combines both core and reaction flow:
\begin{equation*}
U_t = F_t^+ + F_t^- + N^{+}_t + N^{-}_t,
\qquad
S_t = F_t^+ - F_t^- + N^{+}_t - N^{-}_t,
\end{equation*}
where $U_t$ is the unsigned aggregate volume and $S_t$ is total signed order flow.  


\section{Scaling limits of the order flows}
\label{sec:scaling_limits}
In this section, we study the macroscopic behavior of the different order flows introduced above. 
We first address the scaling limits of core orders and then turn to those of reaction orders and the aggregate flow.

\subsection{Scaling limit of the core flow}
\label{sec:scaling_limit_fundamental}
We consider the same model as in Section \ref{sec:model_description}, indexed with the additional parameter $T>0$ that denotes the length of the time interval $[0,T]$ on which the processes are observed. The goal of this section is to establish scaling limits for the core order flow process as $T$ goes to infinity, thus capturing its macroscopic behavior.
Following \cite{jaisson2016rough}, we work in a nearly unstable, heavy–tailed Hawkes regime that captures both the high level of clustering of the core flow and the long memory of trading activity. We formalize this through the following assumptions.

\begin{assumption}
\label{assumption:jaisson:long_memory:1}
There exists a nonnegative sequence $(a_0^T)_{T \geq 0}$ converging to one such that $a_0^T<1$ and 
\begin{equation*}
    \varphi_0^T = a_0^T \varphi_0,
\end{equation*}
for some completely monotone kernel $\varphi_0$ (see \cite{Bernstein1929} for definition) such that $\norm{\varphi_0}_{L^1}=1$. Furthermore, there exists $0 < \alpha_0 < 1$ and a positive constant $K_0$ such that as $t$ tends to infinity,
\begin{equation*}
    \alpha_0 t^{\alpha_0} \int_t^\infty \varphi_0(t) \, dt \to K_0.
\end{equation*}
\end{assumption}
From a probabilistic perspective, a Hawkes process can be viewed as a population process and the norm of the corresponding self-exciting kernel. In this case, $\varphi_0$ can be interpreted as the proportion of descendants in the whole population. In the financial setting, the norm $\norm{\varphi_0}_{L^1}$ can be seen as the proportion of orders that are subsequent to other orders in the market. Most orders fall in this category, in the sense that a large fraction of orders are {\it follower orders}. This is, for example, the case when they are part of the same metaorder or when they are reaction orders to the global flow \cite{hardiman2013critical}. In our model, this translates into the assumption that the norm of the self-exciting kernel converges to one, while remaining strictly below this threshold. In particular, the condition $\norm{\varphi_0}_{L^1} < 1$ ensures the existence of a stationary solution for the intensity. The second part of Assumption~\ref{assumption:jaisson:long_memory:1} imposes a heavy-tailed kernel, which captures strong clustering in order arrivals induced, e.g., due to the splitting of metaorders. Here, we assume a power-law decay, governed by the parameter $\alpha_0$.


To obtain non-degenerate limits for our signed core order flow and unsigned core volume, the parameters $a^T_0$, $\alpha_0$ and the baseline intensity $\nu^T$ of the core flow have to be scaled appropriately: 

\begin{assumption}
\label{assumption:jaisson:long_memory:2}
There exists two constants $\lambda_0, \mu_0 > 0$ such that
\begin{equation*}
        \lim\limits_{T\to\infty} T^{\alpha_0} (1-a^T_0) = \lambda_0 K_0 \frac{\Gamma(1-\alpha_0)}{\alpha_0}
        \qquad \text{and} \qquad
        \lim\limits_{T\to\infty} T^{1-\alpha_0} \nu^T = \mu_0 \frac{\alpha_0}{K_0 \Gamma(1-\alpha_0)},
\end{equation*}
where $\Gamma$ is the Gamma function.
\end{assumption}

Under these assumptions, the long-term average intensity of the Hawkes process $F^{\pm, T}$ is given by $(1-a_0^T)^{-1} \nu^T$. Therefore the average number of trades from $F^{\pm, T}$ on $[0, T]$ scales as $T \nu^T (1-a_0^T)^{-1}$. As a result, it is natural to normalize each of the Hawkes processes by $(1-a_0^T)^{-1} \nu^T T$ and consider the rescaled processes
\begin{equation*}
\wb F_{t}^{\pm, T} = \frac{1-a^T_0}{T\nu^T} F_{tT}^{\pm, T}.
\end{equation*}
For $\alpha_0>0$ and $\lambda_0>0$, we define the function $f^{\alpha_0, \lambda_0}$ by 
\begin{equation*}
f^{\alpha_0,\lambda_0}(x) = \lambda_0 x^{\alpha_0-1} E_{\alpha_0, \alpha_0}(-\lambda_0 x^{\alpha_0}),
\end{equation*}
where $E_{\alpha, \beta}$ is the $(\alpha, \beta)$-Mittag-Leffler function
\begin{equation*}
E_{\alpha, \beta}(x) = \sum_{k=0}^\infty \frac{x^k}{\Gamma(\alpha k + \beta)},
\end{equation*}
see \cite{haubold2011mittag}. 
The following theorem is proved in Appendix~\ref{appendix:proofs}:

\begin{theorem}
\label{thm:cnvg_fundamental}
Under Assumptions \ref{assumption:jaisson:long_memory:1} and \ref{assumption:jaisson:long_memory:2}, the process $(\wb F_{t}^{+, T}, \wb F_{t}^{-, T})_{t \in [0,1]}$ is tight for the Skorokhod topology. Furthermore, any limit point $(F_{t}^{+}, F_{t}^{-})$\footnote{From now on $(F_{t}^{+}, F_{t}^{-})$ denote the limiting processes and no longer the Hawkes processes.}  of $(\wb F_{t}^{+, T}, \wb F_{t}^{-, T})$ satisfies 
\begin{equation*}
    F_{t}^{\pm} = \int_0^t s f^{\alpha_0, \lambda_0}(t-s) \, ds + \frac{1}{\sqrt{\mu_0 \lambda_0}}\int_0^t f^{\alpha_0, \lambda_0}(t-s) Z^{\pm}_s \, ds, 
\end{equation*}
where $Z^{+}$ and $Z^-$ are two independent continuous martingales with quadratic variations $F^{+}$ and $F^-$, respectively.
\end{theorem}

When the kernel's decay parameter satisfies $\alpha_0 > \tfrac12$, the limiting processes identified in Theorem~\ref{thm:cnvg_fundamental} are differentiable, and their derivatives belong to the class of rough Heston–type models developed in \cite{jaisson2016rough, eleuch2019characteristic}. 
By contrast, in the empirically relevant regime $\alpha_0 < \tfrac12$, the core flow displays strong persistence and the limiting processes become non-differentiable. 
Note that in this case, the asymptotic cumulated unsigned volume is a non-differentiable increasing process, which corresponds to the dynamics of the integrated volatility in ``hyper-rough Heston models'' introduced in \cite{jusselin2020noarbitrage}. Determining the exact almost sure Hölder regularity of such processes is a delicate pathwise problem. 
In the present setting, we are able to establish a lower bound for the almost sure Hölder exponent, but a matching upper bound remains out of reach. 
However, Kolmogorov’s continuity theorem and its extensions provide a powerful link between pathwise regularity and moment estimates in \(L^p\). 
For Gaussian processes, these notions coincide \cite{ciesielski1993quelques}; in our case, the limiting processes are not Gaussian, although they share closely related structural features. 
As a consequence, classical Gaussian arguments cannot be applied directly. 
Nevertheless, by exploiting suitable moment estimates, we are able to obtain a sharp characterization of Hölder regularity in the \(L^2\) sense, together with a lower bound for the almost sure Hölder regularity of the sample paths, which we summarize in the following proposition.

\begin{prop}
\label{prop:holder_fundamental}
For any $\varepsilon > 0$, the processes $F^+$ and $F^-$ are almost surely Hölder continuous on $[0,1]$ with exponent $(1 \wedge 2\alpha_0) - \varepsilon$. 
Moreover, in the case $\alpha_0 < \tfrac12$, they are exactly $2\alpha_0$-Hölder continuous in $L^2$, in the sense that there exists a constant $C > 0$ such that, for any $t \in [0,1]$:
\[
\bigl(\E |F_{t+h} - F_t|^2 \bigr)^{1/2}
= C\, h^{2\alpha_0} + o\bigl(h^{2\alpha_0}\bigr), \quad \mbox{as $h \to 0$.}
\]
\end{prop}

As $T$ goes to infinity, the scaled signed core order flow and unsigned core volume satisfy
\begin{equation*}
\wb F_t^{+,T} + \wb F_t^{-,T}
\longrightarrow
F_t^+ + F_t^-
\qquad \text{and} \qquad
\wb F_t^{+,T} - \wb F_t^{-,T}
\longrightarrow
F_t^+ - F_t^-.
\end{equation*}

From Theorem \ref{thm:cnvg_fundamental}, we therefore obtain the following limit theorem:

\begin{proposition}
\label{prop:cnvg_fundamental}
Let
\begin{equation*}
F_t = F_t^+ + F_t^-
\qquad \text{and} \qquad
V_t = F_t^+ - F_t^-
\end{equation*}
denote the scaling limits of the unsigned core volume and signed core flow, respectively, where $F^+$ and $F^-$ are given in Theorem \ref{thm:cnvg_fundamental}. We have
\begin{equation*}
    F_t
    = 2\int_0^t s\,f^{\alpha_0,\lambda_0}(t - s)\,ds
    +
    \frac{1}{\sqrt{\mu_0\,\lambda_0}}
    \int_0^t f^{\alpha_0,\lambda_0}(t - s)\,Z^F_s\,ds
\end{equation*}
and 
\begin{equation*}
    V_t
    =
    \frac{1}{\sqrt{\mu_0\,\lambda_0}}
    \int_0^t f^{\alpha_0,\lambda_0}(t - s)\,Z^V_{s}\,ds,
\end{equation*}
where $Z^F$ and $Z^V$ are two continuous martingales with quadratic variation $F$ and quadratic covariation $V$ such that 
\begin{equation*}
Z^F = Z^+ + Z^- \qquad \text{and} \qquad Z^V = Z^+ - Z^-,
\end{equation*}
where $Z^+$ and $Z^-$ are given in Theorem \ref{thm:cnvg_fundamental}.
\end{proposition}

When $\alpha_0 < \tfrac{1}{2}$, the processes $F$ and $V$ are exactly $2\alpha_0$–Hölder continuous in $L^2$. 
Consequently, the signed core order flow and the unsigned core volume exhibit the same local regularity as a fractional Brownian motion with Hurst exponent 
$$H_0 = 2\alpha_0.$$ 
Moreover, in the high-frequency asymptotic regime relevant for statistical inference, their autocovariance functions coincide with those of a fractional Brownian motion with parameter $H_0$ \cite{chong2022CLT,CDM22,szymanski2026mixed}. 

\subsection{Scaling limit of the reaction orders}
\label{sec:scaling_limit_reaction}
Similarly as for the core flow, we augment the notations for the reaction flow from Section \ref{sec:model_description} with the additional parameter $T$ for the time horizon. We write
\begin{equation*}
\bfLambda_t^T = \int_0^t \bflambda_s^Tds
\end{equation*}
for the compensator of our Hawkes process and the associated martingale is denoted by 
\begin{equation*}
\mathbf{M}_t^T = \mathbf{N}_t^T - \bfLambda_t^T.
\end{equation*}
We are again interested in the macroscopic scaling behavior of the reaction orders. Therefore, in the same spirit as in \cite{eleuch2018microstructural}, we make the following assumption that again reflects the fact that most sent orders can be seen as consequence of some earlier orders. 

\begin{assumption}
\label{assumption:eleuch:1}
There exists a nonnegative sequence $(a_1^T)_{T \geq 0}$ converging to one such that $a_1^T<1$ and
\begin{equation*}
    \bfphi^T = a_1^T \bfphi
\end{equation*}
for some matrix $\bfphi$ whose spectral radius satisfies
\begin{equation*}
\mathcal{S}(\norm{\bfphi}_{L^1}) = \norm{\varphi_1}_{L^1} + \norm{\varphi_2}_{L^1}= 1.
\end{equation*}
\end{assumption}

From \cite{jusselin2020noarbitrage}, we also know that Assumption \ref{assumption:eleuch:1} is necessary in order to obtain non-trivial price impact on the market. 
We write $k_1(t) \geq k_{2}(t)$ for the eigenvalues of $\bfphi(t)$, i.e.,
\begin{equation*}
    k_1(t) = \varphi_1(t) + \varphi_2(t),\quad
    k_2(t) = \varphi_1(t) - \varphi_2(t),
\end{equation*}
and denote by $v_1$, $v_2$ their associated eigenvectors
\begin{equation*}
v_1 = 
\begin{pmatrix}
    1
    \\
    1
\end{pmatrix},
\quad
v_2 = 
\begin{pmatrix}
    1
    \\
    -1
\end{pmatrix}.
\end{equation*}

The following assumption relates to the slowly decreasing behavior of the kernel matrix, that is also necessary to obtain non-trivial market impact, see \cite{jusselin2020noarbitrage}: 

\begin{assumption}
\label{assumption:eleuch:2}
There exists $1/2 < \alpha_1 < 1$ and $K_1 > 0$ such that
\begin{equation*}
    \lim\limits_{t\to\infty} \alpha_1 t^{\alpha_1} \int_t^\infty k_1(s) \, ds  \to K_1.
\end{equation*}
\end{assumption}

We finally need to specify an asymptotic framework similar to that in Assumption \ref{assumption:jaisson:long_memory:2} to ensure our limiting processes are not degenerate. In \cite{eleuch2018microstructural} there is a constant baseline $\mu^T$ and two positive constants $\lambda_1$ and $\mu_1$ such that
\begin{equation*}
    T^{\alpha_1} (1-a^T_1) \to \lambda_1, \qquad \text{and} 
    \qquad
    T^{1-\alpha_1} \mu^T \to \mu_1.
\end{equation*}
However, in our setting the baseline intensity $\boldsymbol{\mu}^T$ is itself stochastic and time‐dependent.  In \cite{eleuch2018microstructural}, $\mu^T$ behaves like $T^{\alpha_1 - 1}$ as $T\to\infty$, so that the expected number of baseline‐driven jumps on $[0,T]$, namely $T\,\mu^T$, grows like $T^{\alpha_1}$. In our case, the number of baseline events between $0$ and $T$ is
$F_T^{+,T} + F_T^{-,T}$,
that is of order $(1-a_0^T)^{-1} T\,\nu^T$.
Therefore, it is natural to replace $T\,\mu^T$ by $T\nu^T(1 - a_0^T)^{-1}$ and to make the following assumption:

\begin{assumption}
\label{assumption:eleuch:3}
There exist $\lambda_1, \mu_1 > 0$ such that
\begin{equation*}
        T^{\alpha_1} (1-a^T_1) \to \lambda_1
        \qquad \text{and} \qquad \frac{T^{1 - \alpha_1} \nu^T}{1 - a_0^T} \to \mu_1.
\end{equation*}
\end{assumption}

By Assumption~\ref{assumption:jaisson:long_memory:2}, the product \(T\nu^T\) is of order \(T^{\alpha_0}\), which implies from Assumption \ref{assumption:eleuch:3} that \(1 - a_0^T\) must be of order \(T^{\alpha_0 - \alpha_1}\).  However, we already have that \(1 - a_0^T\) scales as \(T^{-\alpha_0}\). Hence, to accommodate both of these scalings, we necessarily need
\begin{equation*}
\alpha_1 = 2\,\alpha_0.
\end{equation*}
Note also that from Assumption \ref{assumption:eleuch:2}, we have $1/4 < \alpha_0 < 1/2$. As a consequence, the existence of a nontrivial scaling limit imposes strong structural constraints on the underlying Hawkes model.

In summary, we consider the scaled processes
\begin{align*}
\wb N^{\pm, T}_t =  \frac{(1-a^T_0)(1-a^T_1)}{T\nu^T} N_{tT}^{\pm, T}, &\qquad \qquad  \wb \Lambda_t^{\pm, T} = \frac{(1-a^T_0)(1-a^T_1)}{T\nu^T} \Lambda_{tT}^{\pm, T},
\\
\wb {M}_t^{\pm, T} = &\Big(\frac{(1-a^T_0)(1-a^T_1)}{T\nu^T}\Big)^{1/2} M_{tT}^{\pm, T}.
\end{align*}

We are now ready to state the convergence in distribution of these processes.

\begin{theorem}
\label{thm:cnvg_reaction}
Under Assumptions \ref{assumption:jaisson:long_memory:1}, \ref{assumption:jaisson:long_memory:2}, \ref{assumption:eleuch:1}, \ref{assumption:eleuch:2} and \ref{assumption:eleuch:3}:
\begin{itemize}
    
    \item The process ($\wb N^{+, T}, \wb N^{-, T}, \wb \Lambda^{+, T}, \wb \Lambda^{-, T}, \wb M^{+, T}, \wb M^{-, T}$) is C-tight for the Skorokhod topology. Moreover, each of its limit points ($X, X, X, X, Z^{+}, Z^{-}$) 
    has the rough Heston-type dynamics
    \begin{equation*}
        X_t =  \frac{1}{2}\int_0^t f^{\alpha_1, \lambda_1}(t-s) F_{s} \, ds
        +
        \frac{1}{2\sqrt{\lambda_1 \mu_1}}    \int_0^t f^{\alpha_1, \lambda_1}(t-s)Z_s \, ds,\\
    \end{equation*}
    where $Z = Z^+ + Z^-$, with $Z^+$ and $Z^-$ two continuous martingales with quadratic variation $X$ and zero quadratic covariation, and $F$ is given in Proposition~\ref{prop:cnvg_fundamental}. Furthermore, $X$ behaves as an integrated rough process, and its derivative has Hölder regularity of order $(H_1 - \varepsilon)$ for any $\varepsilon > 0$ on $[0,1]$, where $H_1 = \alpha_1 - 1/2 = H_0 - 1/2$.
    \item The scaled signed reaction flow $ \wb {N}^{+,T} - \wb {N}^{-,T}$ converges in probability to zero.
\end{itemize}
\end{theorem} 

The first part of Theorem \ref{thm:cnvg_reaction} suggests that the roughness of the unsigned volume originates from reaction orders. The second statement shows that under the same rescaling as for the unsigned volume, the signed reaction flow actually vanishes. This means unsigned volume and signed order flow have a different order of magnitude, which will play a crucial role in the study of the asymptotic behaviors of the global flows in the next section. 

\subsection{Scaling limits for the global order flow}
\label{sec:scaling_limit_unsigned}

We now turn to the scaling limits of the aggregate order flow processes 
\begin{equation*}
    \begin{split}
        &U^T_t = F_t^{T,+} + F_t^{-,T} + N^{+,T}_t + N^{-,T}_t,
        \\
        &S^T_t = F_t^{T,+} - F_t^{-,T} + N^{+,T}_t - N^{-,T}_t.
    \end{split}
\end{equation*}
We start with the unsigned volume, which is the easiest case as all the terms have already been fully investigated in Sections \ref{sec:scaling_limit_fundamental} and \ref{sec:scaling_limit_reaction}. We define
\begin{equation*}
\wb U^T_t = \frac{(1-a_0^T)(1-a_1^T)}{T\nu^T}  U^T_{tT},
\end{equation*}
and obtain the following result:
\begin{theorem}
\label{thm:cnvg_scaled_unsigned_flow}
Under Assumptions \ref{assumption:jaisson:long_memory:1}, \ref{assumption:jaisson:long_memory:2}, \ref{assumption:eleuch:1}, \ref{assumption:eleuch:2} and \ref{assumption:eleuch:3}, the scaled unsigned volume $\wb U^T$ is C-tight in the Skorokhod topology. Furthermore if $U$ is a limit point of $\wb U^T$, then $U$ satisfies
 \begin{equation}
 \label{eq:scaling_unsigned_flow}
        U_t = 2X_t =   \int_0^t f^{\alpha_1, \lambda_1}(t-s) F_{s} \, ds
        +
         \sqrt{\frac{1}{\lambda_1 \mu_1}}  \int_0^t f^{\alpha_1, \lambda_1}(t-s)Z_s \, ds,
    \end{equation}
    where $X$ is defined in Theorem \ref{thm:cnvg_reaction}.
\end{theorem}

We see that the contribution of the core flow almost vanishes in the limit of the aggregate unsigned trading volume, which is instead essentially determined by the reaction flow.\footnote{The only trace of the core flow is through the term $F_s$ in the first integral in $U_t$, coming from the baseline intensity of the reaction orders.
}  As a consequence, just like the unsigned reaction volume, the aggregate (cumulative) unsigned volume is an (integrated) rough process in that, for any $\varepsilon>0$, its derivative has Hölder regularity of order $H_1-\varepsilon$ with $H_1= \alpha_1 - 1/2 = H_0 - 1/2$. 

We now turn to the signed order flow. As already observed in Theorem~\ref{thm:cnvg_reaction} for the reaction flow, the same scaling as for the unsigned order flow leads to a trivial limit here:
\begin{equation*}
    \frac{(1-a_0^T)(1-a_1^T)}{T\nu^T}  S^T_{tT} \to 0.
\end{equation*}
The intuition for this is provided by Theorem \ref{thm:cnvg_reaction}: in the reaction flow, the buy and sell order flows have the same asymptotic scaling limits, which implies a vanishing difference. Therefore, we need to adapt the scaling for the signed order flow similarly as in \cite{eleuch2018microstructural}: 
\begin{equation*}
\wb S^T_t = \Big(\frac{(1-a_0^T)(1-a_1^T)}{T\nu^T}\Big)^{1/2} S^T_{tT}.
\end{equation*}

In this regime, we then obtain the following nontrivial limiting result:

\begin{theorem}
\label{thm:cnvg_scaled_signed_flow}
Under Assumptions \ref{assumption:jaisson:long_memory:1}, \ref{assumption:jaisson:long_memory:2}, \ref{assumption:eleuch:1}, \ref{assumption:eleuch:2} and \ref{assumption:eleuch:3}, the scaled signed order flow $\wb S^T$ converges in the sense of finite-dimensional laws to
\begin{equation*}
    S_t =\frac{\sqrt{\lambda_1 \mu_1} (\norm{\varphi_1}_1 -\norm{\varphi_2}_1)}{1-(\norm{\varphi_1}_1-\norm{\varphi_2}_1)} V_t + \frac{1}{1-(\|\varphi_1\|_1-\|\varphi_2\|_1)} (Z^+_t - Z^-_t),
\end{equation*}
where $V$ is given by Proposition \ref{prop:cnvg_fundamental}, and $Z^+$ and $Z^-$ are given by Theorem \ref{thm:cnvg_reaction}.
\end{theorem}

Theorem \ref{thm:cnvg_scaled_signed_flow} shows that the aggregate signed order flow can be decomposed into two distinct components. The first is the contribution of the core order flow. It has the same regularity as a fractional Brownian motion with Hurst exponent $H_0 = 2\alpha_0$ and therefore induces persistence in the aggregate signed flow. The second component is a martingale term, which originates from the reaction orders. 

As discussed in the introduction, this decomposition is crucial to resolving the apparent lack of scale invariance in empirical order flow data. To this end, we replace the complex model from Theorem~\ref{thm:cnvg_scaled_signed_flow} with the simplest process with the same local behavior: a mixed fractional Brownian motion, that is, the sum of a fractional Brownian motion and an independent Brownian motion. Put differently, we use $S_t = W_t + B^{H_0}_t$, where $W$ is a standard Brownian motion, used as a proxy for the reaction driven martingale component, and $B^{H_0}$ is a fractional Brownian motion with Hurst exponent $H_0 = 2\alpha_0$, mirroring the regularity of the fractional component in the scaling limit of the aggregate signed flow.

We can now apply relevant estimators for $H_0$ 
under this mixed fractional Brownian motion approximation as in \cite{CDM22,szymanski2026mixed}. Figure~\ref{fig:single_fbm} shows that when the aggregate signed order flow is approximated by a single fractional Brownian motion, the estimated Hurst exponent depends strongly on the bin size. For very fine sampling the estimate is close to $0.5$, reflecting the dominance of the martingale term originating from reaction orders, while at larger bin sizes the estimate increases steadily as the persistent influence of core orders becomes more pronounced. 
In contrast, Figure~\ref{fig:mixed_fbm_intro} demonstrates that when the flow is modeled as a mixture of a fractional Brownian motion and a Brownian motion, the estimated Hurst exponent stabilizes around $0.65$ across all bin sizes. Therefore, the testable implications of Theorem~\ref{thm:cnvg_scaled_signed_flow} are confirmed by the data.


\section{From order flow to market impact and rough volatility}
\label{sec:links}

In this section, we go on to show that the single parameter $H_0$, not only determines the statistical nature of signed order flow and unsigned volume, but also fixes the shape of the impact function and the roughness of the volatility process. This is done by adapting the arguments of \cite{jusselin2020noarbitrage} to our setting.

The starting point of \cite{jusselin2020noarbitrage} is to assume that prices are martingales and to enforce the absence of statistical arbitrage. Moreover, to rule out profitable roundtrips, permanent price impact is linear~\cite{huberman2004price,gatheral2010no}. Following \cite{jaisson2015market, jusselin2020noarbitrage}, one can then show that the price $P_t$ must satisfy
\begin{equation}
\label{eq:def:price_limit}
    P_t = P_0 + \lim\limits_{s \to \infty} \kappa \, \E[Q^+_s - Q^-_s \mid \mathcal{G}_t],
\end{equation}
where $\kappa$ is the permanent impact coefficient, $Q^+$ and $Q^-$ represent the cumulative buy and sell volumes up to time $t$, respectively, and $(\mathcal{G}_t)_{t \geq 0}$ is the natural filtration generated by the order flows $(Q^+, Q^-)$. Price movements thus correspond to the market’s anticipation of future order flow. This relationship provides a general and model-independent link between order flow dynamics and price evolution, reconciling the strong persistence of order flow with the martingale nature of prices.

If $Q^+$ and $Q^-$ are independent Hawkes processes, then \eqref{eq:def:price_limit} takes the explicit propagator form 
\begin{equation*}
    P_t = P_0 + \kappa \int_0^t \xi(t-s) \, (dQ^+_s - dQ^-_s),
\end{equation*}
where $\xi$ is an explicit kernel compensating the memory of the flow and that can be computed from the Hawkes excitation kernel \cite{jaisson2015market, jusselin2020noarbitrage}. In this setting, it is shown in \cite{jusselin2020noarbitrage} that there exists some $\beta\in (0,1)$ such that the average price deviation at time $t$, $MI(t)$, of a metaorder 
scheduled with a constant trading rate over a renormalized time interval $[0,1]$ satisfies
\begin{equation*}
MI(t) \sim t^{1 - \beta}, \text{ for }t\leq 1,
\end{equation*}
\begin{equation*}
MI(t) \sim t^{1 - \beta}-{(t-1)}^{1 - \beta}, \text{ for }t>1.
\end{equation*}
The parameter $\beta$ is also linked to the tail of the kernel of the Hawkes processes driving the flow: $\phi(t) \sim t^{-(1 + \beta)}$ as $t$ tends to infinity, see \cite{jusselin2020noarbitrage} for details. The celebrated square-root law of market impact corresponds to the case $\beta=1/2$ in the above formulas.\footnote{Note, however, that the exact shape of the relaxation phase is less agreed upon than that of the increasing phase of the impact, see \cite{maitrier2025subtleinterplaysquarerootimpact}.}

Another implication of this framework is that the scaling limit of the price is a rough volatility model. Indeed, the volatility of the price is driven by a rough fractional process with roughness exponent $\beta-1/2$, see again \cite{jusselin2020noarbitrage}. 

In our case, the (asymptotic) signed order flow $Y$ is essentially a mixed fractional Brownian motion, whose fractional component has Hurst exponent $H_0=2\alpha_0$. We see from \eqref{eq:def:price_limit} that the key part of the flow for the link with price dynamics is its predictable part. It is shown in \cite{cheridito2001mixed} that provided $H_0 > 3/4$, $Y$ is a semi-martingale in its natural filtration. Empirically, we find $H_0>3/4$, see Figure \ref{fig:mixed_fbm_intro}. We can therefore decompose $Y$ into an unpredictable martingale component $M$ and a finite-variation component $A$, that is
\begin{equation*}
Y_t = M_t + A_t.
\end{equation*}
To proceed, the key idea is to approximate the finite-variation process $A$ by the difference of two independent Hawkes processes $\wt N^a$ and $\wt N^b$ with same baseline intensity and self-exciting kernel. If this kernel decays as $t^{-(1+\alpha)}$ with $\alpha \in (1/2,1)$, then the scaling limit of $\wt N^a - \wt N^b$ is continuous and its derivative has H\"older regularity of order $\alpha - 1/2 - \varepsilon$ for any $\varepsilon > 0$, see \cite{jaisson2016rough}. For $H_0>3/4$, we know from \cite{chigansky2025smoothdriftmixedfractional} that $A_t$ is differentiable and its derivative has a H\"older regularity of order $2H_0 - 3/2-\varepsilon$ for any $\varepsilon > 0$. Hence the natural choice in the Hawkes approximation is to take
\begin{equation*}
\alpha = 2H_0-1.
\end{equation*}
We therefore obtain the following link between the core order flow, the market impact exponent, and the roughness of price volatility.

\begin{theorem}\label{impact_vol}
Under the previous approximations, we have
\begin{equation*}
MI(t) \sim t^{2-2H_0}, \text{ for }t\leq 1,
\end{equation*}
\begin{equation*}
MI(t) \sim t^{2-2H_0}-{(t-1)}^{2-2H_0}, \text{ for }t>1.
\end{equation*}
Furthermore, the volatility of the price exhibits a rough behavior with a Hurst parameter $2H_0-3/2$.
\end{theorem}

Theorem \ref{impact_vol} provides a structural relation between core order flow memory, market impact shape, and rough volatility. The square-root law corresponds to $H_0=3/4$, implying zero Hurst parameter for the volatility and a roughness exponent of $1/4$ for the unsigned volume. Hence square-root impact ariswes under moderate persistence of the core flow. The reaction flow is essential in this link because it induces the mixed fractional Brownian motion structure. Without reaction flow, a pure core-driven order flow would require $H_0$ close to one to generate a square-root impact.

\begin{rem}
    The volatility appearing in Theorem~\ref{impact_vol} should be interpreted as an extraday volatility. This is because in the applied results from \cite{jusselin2020noarbitrage}, the authors establish a connection between volatility roughness and the tail behavior of the underlying Hawkes processes by studying the scaling limit of price dynamics. Therefore, the volatility under consideration here corresponds to a time scale that is long enough for prices to exhibit diffusive behavior. Extraday volatilities are known to have Hurst exponents between $0.05$ and $0.15$, which agrees well with our results for the value of the Hurst parameter of the volatility. Note however that forthcoming studies show that volatility is also rough at the intraday scale \cite{chong2026intraday}, with Hurst exponents around $0.25$, so larger than typical values for extraday volatility. This is in line with our findings for the intraday unsigned order flow, which is obviously tightly linked to intraday volatility, and has a Hurst exponent larger than the one of the extraday volatility in our model. 
\end{rem}

\section{Conclusion}
This paper develops a unified modeling framework for the joint dynamics of signed order flow, unsigned volume, market impact and volatility. This allows us to capture salient empirical properties of all of these quantities with a single structural parameter, inherited from the persistence of the core order flow in the model's microfoundation:
\begin{itemize}
\item[$\bullet$] 
\textbf{Persistent signed order flow.}  
The signed order flow is a mixed fractional process, with diffusive behavior at very high-frequency and persistence $H_0$ emerging at larger scales. 

\item[$\bullet$] 
\textbf{Power-law market impact scaling.}  
The average price response to a large order follows a power law with exponent $2-2H_0$.

\item[$\bullet$] 
\textbf{Rough volatility.}  
(Extraday) volatility sample paths are rough, with Hurst exponent $2H_0-3/2$.

\item[$\bullet$] 
\textbf{Rough traded volume.}  
Unsigned traded volume exhibits a rough structure too, with Hurst parameter $H_0-1/2$, close to intraday volatility dynamics.
\end{itemize}

In particular, with the values of $H_0 \approx 0.75$--$0.8$ esimtated from signed order flow data, this model consistently reproduce all the stylized facts mentioned in the introduction, in line with Figures~\ref{fig:mixed_fbm_intro} and \ref{fig:unsigned_volume_hurst}. Interestingly, this provides fresh mathematical and econometric support for the view that {\it financial markets are at the edge of criticality}. To wit, prices are diffusive (Hurst parameter very near 1/2 to preclude arbitrage), the Hurst parameter for volatility is close to zero, and the mixed-fractional order flow is just about a semimartingale (in that its fractional part has a  Hurst parameter just above $3/4$). 

On a less technical level, our findings highlight that despite the proliferation of reactive trading, the structure of financial markets continues to be governed by the slow, persistent rhythm of the core order flow.



\bibliographystyle{plain}
\bibliography{library_submitted}

@article{carmona.webster.19,
  title={The self-financing equation in limit order book markets},
  author={Carmona, Ren\'e and Webster, Kevin},
  journal={Finance and Stochastics},
  volume={23},
  number={3},
  pages={729--759},
  year={2019},
}

@article{guasoni.weber.17,
  title={Dynamic trading volume},
  author={Guasoni, Paolo and Weber, Marko},
  journal={Mathematical Finance},
  volume={27},
  number={2},
  pages={313--349},
  year={2017},
}

@unpublished{chigansky2025smoothdriftmixedfractional,
      title={How smooth is the drift of the mixed fractional {B}rownian motion?}, 
      author={Pavel Chigansky and Marina Kleptsyna},
      year={2025},
     note={Preprint},
      journal={https://arxiv.org/abs/2511.22542}
}

@article{cheridito2001mixed,
  author  = {Cheridito, Patrick},
  title   = {{Mixed fractional {B}rownian motion}},
  journal = {Bernoulli},
  year    = {2001},
  volume  = {7},
  number  = {6},
  pages   = {913--934},
  doi     = {10.3150/bj/1199406791}
}

@incollection{bouchaud2008marketsslowlydigestchange,
  title={How markets slowly digest changes in supply and demand},
  author={Bouchaud, Jean-Philippe and Farmer, Doyne and Lillo, Fabrizio},
  booktitle={Handbook of financial markets: dynamics and evolution},
  pages={57--160},
  year={2009},
  publisher={Elsevier}
}

@unpublished{maitrier2025subtleinterplaysquarerootimpact,
  title={The Subtle Interplay between Square-root Impact, Order Imbalance \& Volatility: A Unifying Framework},
  author={Maitrier, Guillaume and Bouchaud, Jean-Philippe},
  note={Preprint},
  year={2025}
}

@article{bachelier1900theorie,
    title = {Th{\'e}orie de la sp{\'e}culation},
    author = {Bachelier, Louis},
    Journal = {Annales Scientifiques de l'{\'E}cole Normale Sup{\'e}rieure},
    volume = {17},
    pages = {21--86},
    year = {1900}
}

@article{Sato_2025,
   title={Strict Universality of the Square-Root Law in Price Impact across Stocks: A Complete Survey of the {T}okyo Stock Exchange},
   volume={135},
   ISSN={1079-7114},
   url={http://dx.doi.org/10.1103/65jz-81kv},
   DOI={10.1103/65jz-81kv},
   number={25},
   journal={Physical Review Letters},
   publisher={American Physical Society (APS)},
   author={Sato, Yuki and Kanazawa, Kiyoshi},
   year={2025},
   month=dec }

@article{li2025weak,
  title={Weak identification of long memory with implications for volatility modeling},
  author={Li, Jia and Phillips, Peter {C. B.} and Shi, Shuping and Yu, Jun},
  journal={Review of Financial Studies},
  volume={38},
  pages={3117-3148},
  year={2025},
  publisher={Oxford University Press}
}

@article{chong2025nonparametric,
  title={A nonparametric test for rough volatility},
  author={Chong, Carsten and Todorov, Viktor},
  journal={Journal of the American Statistical Association},
volume={120},
  number={552},
  pages={2772--2783},
  year={2025},
  publisher={Taylor \& Francis}
}

@article{bayer2016pricing,
    author = {Bayer, Christian and Friz, Peter K. and Gatheral, Jim},
    journal = {{Quantitative Finance}},
    number = {6},
    pages = {887--904},
    publisher = {Taylor \& Francis},
    title = {{P}ricing under rough volatility},
    volume = {16},
    year = {2016}
}

@article{bennedsen2022decoupling,
    author = {Bennedsen, Mikkel and Lunde, Asger and Mikko S. Pakkanen},
    journal = {Journal of Financial Econometrics},
    number = {5},
    pages = {961--1006},
    publisher = {Oxford University Press},
    title = {{D}ecoupling the short and long term behavior of stochastic volatility},
    volume = {20},
    year = {2022}
}

@article{benzaquen2018market,
    title = {Market impact with multi-timescale liquidity},
    author = {Benzaquen, Michael and Bouchaud, Jean-Philippe},
    journal = {Quantitative Finance},
    volume = {18},
    number = {11},
    pages = {1781--1790},
    year = {2018},
    publisher = {Taylor \& Francis}
}

@article{bershova2013non,
    author = {Bershova, Nataliya and Rakhlin, Dmitry},
    journal = {{Quantitative Finance}},
    number = {11},
    pages = {1759--1778},
    publisher = {Taylor \& Francis},
    title = {{T}he non-linear market impact of large trades: {E}vidence from buy-side order flow},
    volume = {13},
    year = {2013}
}

@book{billingsley1968convergence,
    author = {Billingsley, Patrick},
    publisher = {Wiley-Interscience},
    title = {{C}onvergence of {P}robability {M}easures},
    year = {1968}
}

@article{black1973pricing,
    author = {Fischer Black and Myron Scholes},
    issn = {00223808, 1537534X},
    journal = {{Journal of Political Economy}},
    number = {3},
    pages = {637--654},
    publisher = {University of Chicago Press},
    title = {The Pricing of Options and Corporate Liabilities},
    url = {http://www.jstor.org/stable/1831029},
    volume = {81},
    year = {1973}
}

@article{han2025rate,
  title={On the rate of convergence of estimating the {H}urst parameter of rough stochastic volatility models},
  author={Han, Xiyue and Schied, Alexander},
 journal={SIAM Journal on Financial Mathematics},
volume={16},
  number={4},
  pages={1336--1349},
  year={2025}
}

@article{shi2025fractional,
  title={Fractional {G}aussian noise: Spectral density and estimation methods},
  author={Shi, Shuping and Yu, Jun and Zhang, Chen},
  journal={Journal of Time Series Analysis},
  volume={46},
  number={6},
  pages={1146--1174},
  year={2025},
  publisher={Wiley Online Library}
}

@article{dayri2015large,
  title={Large tick assets: implicit spread and optimal tick size},
  author={Dayri, Khalil and Rosenbaum, Mathieu},
  journal={Market Microstructure and Liquidity},
  volume={1},
  number={01},
  pages={1550003},
  year={2015},
  publisher={World Scientific}
}

@article{bolko2023gmm,
    title = {{A GMM approach to estimate the roughness of stochastic volatility}},
    author = {Bolko, Anine E. and Christensen, Kim and Pakkanen, Mikko S. and Veliyev, Bezirgen},
    journal = {{Journal of Econometrics}},
    volume = {235},
    number = {2},
    pages = {745--778},
    year = {2023},
    publisher = {Elsevier}
}

@article{CDM22,
    author = {Chong, Carsten and Delerue, Thomas and Mies, Fabian},
    title = {Rate-optimal estimation of mixed semimartingales},
    journal = {Annals of Statistics},
    volume = {153},
    number = {1},
    pages = {219--244},
    year = {2025}
}

@article{chong2022CLT,
    title={{Statistical inference for rough volatility: Central limit theorems}},
    author={Chong, Carsten and Hoffmann, Marc and Liu, Yanghui and Rosenbaum, Mathieu and Szymanski, Gr{\'e}goire},
    journal={{Annals of Applied Probability}},
    volume={34},
    number={3},
    pages={2600--2649},
    year={2024},
    publisher={Institute of Mathematical Statistics}
}

@article{chong2022minimax,
  title={{Statistical inference for rough volatility: Minimax theory}},
  author={Chong, Carsten and Hoffmann, Marc and Liu, Yanghui and Rosenbaum, Mathieu and Szymanski, Gr{\'e}goire},
  journal={{The Annals of Statistics}},
  volume={52},
  number={4},
  pages={1277--1306},
  year={2024},
  publisher={Institute of Mathematical Statistics}
}

@article{ciesielski1993quelques,
    author = {Ciesielski, Zbigniew and Kerkyacharian, G{\'e}rard and Roynette, Bernard},
    journal = {Studia Mathematica},
    pages = {171--204},
    publisher = {Instytut Matematyczny Polskiej Akademii Nauk},
    title = {{Q}uelques espaces fonctionnels associ{\'e}s {\`a} des processus gaussiens},
    volume = {107},
    year = {1993}
}

@article{comte1996long,
  title={Long memory continuous time models},
  author={Comte, Fabienne and Renault, Eric},
  journal={Journal of Econometrics},
  volume={73},
  number={1},
  pages={101--149},
  year={1996},
  publisher={Elsevier}
}

@article{donier2015fully,
        title = {A fully consistent, minimal model for non-linear market impact},
        author = {Donier, Jonathan and Bonart, Julius and Mastromatteo, Iacopo and Bouchaud, Jean-Philippe},
        journal = {Quantitative Finance},
        volume = {15},
        number = {7},
        pages = {1109--1121},
        year = {2015},
        publisher = {Taylor \& Francis}
    }

@article{eleuch2018microstructural,
    author = {El Euch, Omar and Fukasawa, Masaaki and Rosenbaum, Mathieu},
    journal = {{Finance and Stochastics}},
    number = {2},
    pages = {241--280},
    publisher = {Springer},
    title = {{T}he microstructural foundations of leverage effect and rough volatility},
    volume = {22},
    year = {2018}
}

@article{eleuch2019characteristic,
    author = {El Euch, Omar and Rosenbaum, Mathieu},
    journal = {{Mathematical Finance}},
    number = {1},
    pages = {3--38},
    publisher = {Wiley Online Library},
    title = {{T}he characteristic function of rough {H}eston models},
    volume = {29},
    year = {2019}
}

@article{farmer2006market,
    title = {Market efficiency and the long-memory of supply and demand: Is price impact variable and permanent or fixed and temporary?},
    author = {Farmer, Doyne and Gerig, Austin and Lillo, Fabrizio and Mike, Szabolcs},
    journal = {Quantitative Finance},
    volume = {6},
    number = {02},
    pages = {107--112},
    year = {2006},
    publisher = {Taylor \& Francis}
}

@article{farmer2013efficiency,
    author = {Farmer, Doyne and Gerig, Austin and Lillo, Fabrizio and Waelbroeck, Henri},
    journal = {{Quantitative Finance}},
    number = {11},
    pages = {1743--1758},
    publisher = {Taylor \& Francis},
    title = {{H}ow efficiency shapes market impact},
    volume = {13},
    year = {2013}
}

@article{gatheral2010no,
    author = {Gatheral, Jim},
    journal = {{Quantitative Finance}},
    number = {7},
    pages = {749--759},
    publisher = {Taylor \& Francis},
    title = {{N}o-dynamic-arbitrage and market impact},
    volume = {10},
    year = {2010}
}

@article{gatheral2018volatility,
    author = {Gatheral, Jim and Jaisson, Thibault and Rosenbaum, Mathieu},
    journal = {{Quantitative Finance}},
    number = {6},
    pages = {933--949},
    publisher = {Taylor \& Francis},
    title = {{V}olatility is rough},
    volume = {18},
    year = {2018}
}

@article{hardiman2013critical,
    title = {{Critical reflexivity in financial markets: a Hawkes process analysis}},
    author = {Hardiman, Stephen J. and Bercot, Nicolas and Bouchaud, Jean-Philippe},
    journal = {The European Physical Journal B},
    volume = {86},
    pages = {1--9},
    year = {2013},
    publisher = {Springer}
}

@Article{haubold2011mittag,
journal={Journal of Applied Mathematics},
author={Hans J. Haubold and Arakaparampil M. Mathai and Ram K. Saxena},
title={{M}ittag-{L}effler Functions and Their Applications},
year={2011},
pages={1-51},
volume={2011},
doi={10.1155/2011/298628},
url={https://ideas.repec.org/a/hin/jnljam/298628.html},
}

@article{kyle2016market,
  title={Market microstructure invariance: Empirical hypotheses},
  author={Kyle, Albert S. and Obizhaeva, Anna A.},
  journal={Econometrica},
  volume={84},
  number={4},
  pages={1345--1404},
  year={2016},
  publisher={Wiley Online Library}
}

@article{wang2024optimal,
  title={On the optimal forecast with the fractional {B}rownian motion},
  author={Wang, Xiaohu and Yu, Jun and Zhang, Chen},
  journal={Quantitative Finance},
  volume={24},
  number={2},
  pages={337--346},
  year={2024},
  publisher={Taylor \& Francis}
}

@article{horst2022microstructure,
  title={The microstructure of stochastic volatility models with self-exciting jump dynamics},
  author={Horst, Ulrich and Xu, Wei},
  journal={The Annals of Applied Probability},
  volume={32},
  number={6},
  pages={4568--4610},
  year={2022},
  publisher={Institute of Mathematical Statistics}
}

@article{huberman2004price,
    author = {Huberman, Gur and Stanzl, Werner},
    journal = {Econometrica},
    number = {4},
    pages = {1247--1275},
    publisher = {Wiley Online Library},
    title = {{P}rice manipulation and quasi-arbitrage},
    volume = {72},
    year = {2004}
}

@book{jacod1987limit,
    author = {Jacod, Jean and Shiryaev, Albert N.},
    publisher = {Springer, Berlin, Heidelberg},
    title = {{Limit Theorems for Stochastic Processes}},
    year = {1987}
}

@unpublished{maitrier2025subtle2,
  title={The Subtle Interplay between Square-root Impact, Order Imbalance \& Volatility {II}: An Artificial Market Generator},
  author={Maitrier, Guillaume and Loeper, Gr{\'e}goire and Bouchaud, Jean-Philippe},
  note={Preprint},
  year={2025}
}

@unpublished{naviglio2025estimation,
  title={Why is the estimation of metaorder impact with public market data so challenging?},
  author={Naviglio, Manuel and Bormetti, Giacomo and Campigli, Francesco and Rodikov, German and Lillo, Fabrizio},
  note={Preprint},
  year={2025}
}

@article{sato2023inferring,
  title={Inferring microscopic financial information from the long memory in market-order flow: A quantitative test of the {L}illo-{M}ike-{F}armer model},
  author={Sato, Yuki and Kanazawa, Kiyoshi},
  journal={Physical Review Letters},
  volume={131},
  number={19},
  pages={197401},
  year={2023},
  publisher={APS}
}

@article{jaisson2015limit,
    author = {Jaisson, Thibault and Rosenbaum, Mathieu},
    journal = {{The Annals of Applied Probability}},
    number = {2},
    pages = {600--631},
    publisher = {Institute of Mathematical Statistics},
    title = {{L}imit theorems for nearly unstable {H}awkes processes},
    volume = {25},
    year = {2015}
}

@article{jaisson2015market,
    author = {Jaisson, Thibault},
    journal = {{Quantitative Finance}},
    number = {7},
    pages = {1123--1135},
    publisher = {Taylor \& Francis},
    title = {{M}arket impact as anticipation of the order flow imbalance},
    volume = {15},
    year = {2015}
}

@article{jaisson2016rough,
    author = {Jaisson, Thibault and Rosenbaum, Mathieu},
    journal = {{The Annals of Applied Probability}},
    number = {5},
    pages = {2860--2882},
    publisher = {Institute of Mathematical Statistics},
    title = {Rough fractional diffusions as scaling limits of nearly unstable heavy tailed {H}awkes processes},
    volume = {26},
    year = {2016}
}

@article{jusselin2020noarbitrage,
    title = {No-arbitrage implies power-law market impact and rough volatility},
    author = {Jusselin, Paul and Rosenbaum, Mathieu},
    journal = {{Mathematical Finance}},
    volume = {30},
    number = {4},
    pages = {1309--1336},
    year = {2020},
    publisher = {Wiley Online Library}
}

@article{karpoff.87,
  title={The relation between price changes and trading volume: A survey},
  author={Karpoff, Jonathan M.},
  journal={Journal of Financial and Quantitative Analysis},
  volume={22},
  number={1},
  pages={109--126},
  year={1987},
}

@article{lo1991long,
  title={Long-term memory in stock market prices},
  author={Lo, Andrew W.},
  journal={Econometrica},
  pages={1279--1313},
  year={1991},
volume={59},
number={5},
pages={1279--1313},
}

@article{loeb.83,
  title={Trading cost: the critical link between investment information and results},
  author={Loeb, Thomas F.},
  journal={Financial Analysts Journal},
  volume={39},
  number={3},
  pages={39--44},
  year={1983},
}

@unpublished{frazzini.al.18,
  title={Trading costs},
  author={Frazzini, Andrea and Israel, Ronen and Moskowitz, Tobias J.},
  note={Preprint},
 year={2018},
}

@article{lo.wang.00,
  title={Trading volume: definitions, data analysis, and implications of portfolio theory},
  author={Lo, Andrew W. and Wang, Jiang},
  journal={Review of Financial Studies},
  volume={13},
  number={2},
  pages={257--300},
  year={2000},
}

@article{li2016generalized,
  title={Generalized method of integrated moments for high-frequency data},
  author={Li, Jia and Xiu, Dacheng},
  journal={Econometrica},
  volume={84},
  number={4},
  pages={1613--1633},
  year={2016},
  publisher={Wiley Online Library}
}

@article{lillo2004long,
    author = {Lillo, Fabrizio and Farmer, Doyne},
    journal = {Studies in Nonlinear Dynamics \& Econometrics},
    number = {3},
    publisher = {De Gruyter},
    title = {{T}he long memory of the efficient market},
    volume = {8},
    year = {2004}
}

@article{madhavan1997security,
  title={{Why do security prices change? A transaction-level analysis of NYSE stocks}},
  author={Madhavan, Ananth and Richardson, Matthew and Roomans, Mark},
  journal={Review of Financial Studies},
  volume={10},
  number={4},
  pages={1035--1064},
  year={1997},
  publisher={Oxford University Press}
}

@article{wyart2008relation,
  title={Relation between bid--ask spread, impact and volatility in order-driven markets},
  author={Wyart, Matthieu and Bouchaud, Jean-Philippe and Kockelkoren, Julien and Potters, Marc and Vettorazzo, Michele},
  journal={Quantitative Finance},
  volume={8},
  number={1},
  pages={41--57},
  year={2008},
  publisher={Taylor \& Francis}
}

@article{toth2011anomalous,
    author = {T{\'o}th, Bence and Lemperiere, Yves and Deremble, Cyril and De Lataillade, Joachim and Kockelkoren, Julien and Bouchaud, Jean-Philippe},
    journal = {Physical Review X},
    number = {2},
    pages = {021006},
    publisher = {APS},
    title = {{A}nomalous price impact and the critical nature of liquidity in financial markets},
    volume = {1},
    year = {2011}
}

@book{Gripenberg1990,
  author    = {Gripenberg, Gustaf and Londen, Stig-Olof and Staffans, Olof},
  title     = {Volterra Integral and Functional Equations},
  year      = {1990},
  publisher = {Cambridge University Press},
  series    = {Encyclopedia of Mathematics and its Applications},
  volume    = {34},
  address   = {Cambridge},
  doi       = {10.1017/CBO9780511662805}
}

@article{Bernstein1929,
  author  = {Bernstein, Serge},
  title   = {Sur les fonctions absolument monotones},
  journal = {Acta Mathematica},
  year    = {1929},
  volume  = {52},
  pages   = {1--66},
  doi     = {10.1007/BF02592629}
}

@article{BouchaudEtAl2004Fluctuations,
  author  = {Jean-Philippe Bouchaud and Yuval Gefen and Marc Potters and Matthieu Wyart},
  title   = {Fluctuations and response in financial markets: the subtle nature of ‘random’ price changes},
  journal = {Quantitative Finance},
  year    = {2004},
  volume  = {4},
  number  = {2},
  pages   = {176--190},
  doi     = {10.1088/1469-7688/4/2/006}
}

@article{LilloMikeFarmer2005,
  author  = {Fabrizio Lillo and Szabolcs Mike and J. Doyne Farmer},
  title   = {Theory for long memory in supply and demand},
  journal = {Physical Review E},
  year    = {2005},
  volume  = {71},
  pages   = {066122},
  doi     = {10.1103/PhysRevE.71.066122}
}

@unpublished{chong2026intraday,
  title  = {Intraday Volatility Dynamics},
  author = {Chong, Carsten and Hoffmann, Marc and Rosenbaum, Mathieu and Szymanski, Gr{\'e}goire},
  year   = {2026},
  note   = {Preprint}
}

@unpublished{szymanski2026mixed,
  title={Asymptotic efficiency for mixed fractional {B}rownian motion},
  author={Szymanski, Gr{\'e}goire and Takabatake, Tetsuya},
  note   = {Working paper},
  year={2026}
}

@article{chahdi2024theory,
  title={A theory of passive market impact},
  author={Ouazzani Chahdi, Youssef  and Rosenbaum, Mathieu and Szymanski, Gr{\'e}goire},
  Journal={Finance and Stochastics},
volume  = {to appear},
  year={2026}
}

\appendix

\section{Useful results about Hawkes processes}

In this section, we summarize some useful results about Hawkes processes with time-varying baseline. The proofs are omitted for conciseness. They can however be easily adapted from the constant baseline case, see for instance \cite{eleuch2019characteristic, chahdi2024theory}.

\begin{definition}
A Hawkes process with baseline (or background rate) $\mu : [0,\infty) \to [0,\infty)$ and self-exciting kernel $\varphi : [0,\infty) \to \mathbb{R}$ is a process $N$ adapted to some filtration $(\mathcal{F}_t)_t$ such that the compensator $\Lambda$ of $N$ has the form $\Lambda_t = \int_0^t \lambda_s\, ds$ where
\begin{equation*}
    \lambda_t = \mu_t + \int_0^{t-} \varphi (t-s) \,dN_s.
\end{equation*}
\end{definition}

\begin{lemma}
\label{lemma:expressions_M}
Define $M = N-\Lambda$ and $\psi=\sum_{k \geq 1} \varphi^{* k}$ where $\varphi^{* k}$ stands for the $k$-fold convolution of $\varphi$. Then for any $0 \leq t \leq T$, we have
\begin{align*}
    \lambda_t &= \mu_t + \int_0^t \psi(t-s) \mu_s \, ds + \int_0^{t-} \psi(t-s) dM_s,
    \\
    \int_0^t \lambda_s \,ds 
    &=
    \int_0^t \mu_s \, ds
    +\int_0^t \psi(t-s) \int_0^s \mu_u \, du \,ds
    +\int_0^t \psi(t-s) M_s \, ds.
\end{align*}
\end{lemma}

\begin{lemma}
\label{lemma:m1}
For any $0 \leq t \leq T$, we have
\begin{align*}
    \mathbb{E}[\lambda_t] &= \mu_t + \int_0^t \psi(t-s) \mu_s \, ds.
\end{align*}
\end{lemma}

\section{Proof of the results of Section \ref{sec:scaling_limits}}
\label{appendix:proofs}

\subsection{Proof of Theorem \ref{thm:cnvg_fundamental}}
\label{proof:cnvg_fundamental}
Consider a standard Hawkes process $N^T$ with same baseline intensity $\nu^T$ and kernel $\varphi_0^T$ as $ F^{\pm,T}$. We then define
\begin{equation*}
\begin{split}
    &\wb N^T_t = \frac{1-a_0^T}{T\nu^T} N^{T}_{tT},
    \\
    &\wb \Lambda^T_t = \frac{1-a^T_0}{T\nu^T} \Lambda^{T}_{tT},
    \\
    &\wb M^T_t = \Big(\frac{1-a^T_0}{T\nu^T}\Big)^{1/2} M^{T}_{tT}.
\end{split}
\end{equation*}
The proof is then split into five parts:
\begin{itemize}
\item Step 1: We show that the sequence $(\wb \Lambda^T)$ is C-tight.
\item Step 2: We show that the sequences of martingales $(\wb X^T - \wb \Lambda^T)$ tends to zero in probability, uniformly on $[0,1]$.
\item Step 3: Under Assumptions \ref{assumption:jaisson:long_memory:1} and \ref{assumption:jaisson:long_memory:2}, the sequence $(\wb M^T, \wb X^T)$ is tight. Furthermore, if $(Z, X)$ is a limit point of $(\wb M^T, \wb X^T)$, then $Z$ is a continuous martingale and $[Z, Z]=X$.
\item Step 4: We conclude the convergence of the process $(\wb N^T_t, \wb \Lambda^T_t, \wb M^T_t)$ in distribution for the Skorokhod topology towards $(X,X,Z)$ where $X$ and $Z$ are given in Theorem \ref{thm:cnvg_fundamental}.
\item Step 5: We prove the Hölder property for $X$.
\end{itemize}
In this paper, we only prove that $\wb \Lambda^T$ is tight; the remaining steps can be found in the proof of Theorem 3.1 in \cite{jaisson2016rough}. Let us now prove the following lemma.   
\begin{lemma}
The sequence $(\wb \Lambda^T)$ is C-tight.
\end{lemma}
\begin{proof}
Let $\psi_0^T=\sum_{k \geq 1} (\varphi_0^T)^{* k}$. We know from Lemma \ref{lemma:m1} and Assumption \ref{assumption:jaisson:long_memory:1} that
\begin{equation*}
\mathbb{E}[\lambda^T_t] = \nu^T + \int_0^t \psi_0^T(t-s) \nu^T \, ds \leq \nu^T(1 + \norm{\psi_0^T}_1) \leq \frac{\nu^T}{1 - a_0^T}\end{equation*}
and from Assumption \ref{assumption:jaisson:long_memory:1} that $\norm{\psi_0^T}_1 = (1 - a_0^T)^{-1}a_0^T$.
This implies
\begin{equation*}
    \frac{1-a^T_0}{\nu^T} \sup_t \E[\lambda_t^T] \leq 1
\end{equation*}
and therefore 
\begin{equation*}
    \E[\wb X_1^T] = \E[\wb \Lambda_1^T] \leq 1.
\end{equation*}
Moreover, since 
\begin{equation*}
    \big<\wb M^T_t,\wb M^T_t\big> = \wb \Lambda^T_t
\end{equation*}
the Burkholder-Davis-Gundy inequality then ensures
\begin{equation*}
    \E[\sup_{t \leq 1} |\wb M_t^{T}|^2] \leq C
\end{equation*} 
for a constant $C > 0$. We now prove the tightness of $\wb \Lambda^T$. We write
\begin{align*}
    \wb \Lambda^T_t 
    &= \frac{1-a^T_0}{T\nu^T}
    \Big(
    \nu^TtT 
    + 
    \int_0^{tT} \psi_0^T(tT -s)s \, ds \nu^T 
    + 
    \int_0^{tT} \psi_0^T(tT-s) M^T_s \, ds
    \Big)
    \\
    &=
    \Big(
    (1-a^T_0)t
    + 
    T(1-a^T_0)\int_0^{tT} \psi_0^T(T(t-s)) s\, ds
    \Big)
    +
    \frac{1-a^T_0}{T\nu^T}
    \int_0^{t} T\psi_0^T(T(t-s)) M^T_{Ts} \, ds.
\end{align*}
The authors in \cite{jaisson2016rough} prove the uniform convergence of the first term towards the process
\begin{equation*}
\int_0^t sf^{\alpha_0, \lambda_0}(t-s)ds,
\end{equation*}
and therefore is tight. We then focus on the second one and we set
\begin{align*}
    \wt \Lambda^T_t 
    &=
    \frac{1-a^T_0}{T\nu^T}
    \int_0^{t} T\psi_0^T(T(t-s)) M^T_{Ts} \, ds
    \\
    &=
    \Big(\frac{1-a^T_0}{T\nu^T}\Big)^{1/2}
    \int_0^{t} T\psi_0^T(T(t-s)) \wb M^T_{s} \, ds
    \\
    &=
    \Big(\frac{1}{(1-a^T_0)T\nu^T}\Big)^{1/2}
    \int_0^{t} \rho_0^T(t-s) \wb M^T_{s} \, ds
\end{align*}
with
\begin{equation*}
    \rho_0^T(t) = (1-a^T_0)T\psi_0^T(Tt).
\end{equation*}
To prove the tightness of $\wt \Lambda^T$, we use Theorem 7.3. in \cite{billingsley1968convergence} which states that $\wt \Lambda^T$ is tight provided the following two conditions hold:
\begin{itemize}
    \item For each  $\eta > 0$, there exist $a > 0$ such that
    \begin{equation*}
        \limsup_{T} \PX(|\wt \Lambda_0^T| \geq a) \leq \eta.
    \end{equation*}
    \item For each $\varepsilon > 0$, we have
    \begin{equation*}
        \lim_{\delta\to0}\limsup_{T} \PX(\omega(\wt \Lambda^T; \delta) \geq \varepsilon) = 0
    \end{equation*}
    where we use the notation
    \begin{equation*}
        \omega(x; \delta) = \sup_{|t-s|\leq\delta, 0 \leq s \leq t \leq 1} |x(t) - x(s)|.
    \end{equation*}
\end{itemize}

for $\delta > 0$. The first condition clearly holds. 
We prove that $\wt \Lambda^T$ verifies the second condition. We first write for $0 \leq s \leq t \leq s+\delta \leq 1$
\begin{align*}
    |\wt \Lambda_t^T - \wt \Lambda_s^T|
    &=
    \Big|
    \Big(\frac{1}{(1-a^T_0)T\nu^T}\Big)^{1/2}
    \int_0^{t} \rho_0^T(t-u) \wb M^T_{u} \, du
    -
    \Big(\frac{1}{(1-a^T_0)T\nu^T}\Big)^{1/2}
    \int_0^{s} \rho_0^T(s-u) \wb M^T_{u} \, du
    \Big|
    \\
    &=
    \Big(\frac{1}{(1-a^T_0)T\nu^T}\Big)^{1/2}
    \Big|
    \int_s^{t} \rho_0^T(t-u) \wb M^T_{u} \, du
    +
    \int_0^{s} (\rho_0^T(t-u)-\rho_0^T(s-u)) \wb M^T_{u} \, du
    \Big|
    \\
    &\leq
    \Big(\frac{1}{(1-a^T_0)T\nu^T}\Big)^{1/2}
    \Big(
    \int_s^{t} \rho_0^T(t-u) \, du
    +
    \int_0^{s} |\rho_0^T(t-u)-\rho_0^T(s-u)| \, du
    \Big)
    \sup_{u \leq 1} |\wb M^T_{u}|.
\end{align*}
Under Assumption \ref{assumption:jaisson:long_memory:1}, the kernel $\varphi_0$ is completely monotone and it follows from Theorem 5.4 in \cite{Gripenberg1990} that $\rho_0^T$ is decreasing. Since $|t-s|\leq\delta$, we get
\begin{align*}
    |\wt \Lambda_t^T - \wt \Lambda_s^T|
    &\leq
    \Big(\frac{1}{(1-a^T_0)T\nu^T}\Big)^{1/2}
    \Big(
    \int_0^{\delta} \rho_0^T(u) \, du
    +
    \int_0^{s} |\rho_0^T(t-s + u)-\rho_0^T(u)| \, du
    \Big)
    \sup_{u \leq 1} |\wb M^T_{u}|
    \\
    &\leq
    \Big(\frac{1}{(1-a^T_0)T\nu^T}\Big)^{1/2}
    \Big(
    \int_0^{\delta} \rho_0^T(u) \, du
    +
    \int_0^{s} \rho_0^T(u) \, du
    -
    \int_{t-s}^{t} \rho_0^T(u) \, du
    \Big)
    \sup_{u \leq 1} |\wb M^T_{u}|
    \\
    &\leq
    2
    \Big(\frac{1}{(1-a^T_0)T\nu^T}\Big)^{1/2}
    \int_0^{\delta} \rho_0^T(u) \, du
    \sup_{u \leq 1} |\wb M^T_{u}|.
\end{align*}
Using Markov's inequality, we deduce that
\begin{equation*}
    \PX(\omega(\wt \Lambda^T; \delta) \geq \varepsilon) 
    \leq
    2
    \varepsilon^{-1}
    \Big(\frac{1}{(1-a^T_0)T\nu^T}\Big)^{1/2}
    \int_0^{\delta} \rho_0^T(u) \, du
    \,
    \E[\sup_{u\leq 1} |\wb M^T_{u}|]
    \leq
    C'
    \int_0^{\delta} \rho_0^T(u) \, du
\end{equation*}
for some positive constant $C'$ and we conclude using
\begin{equation*}
    \lim_{\delta\to0}\limsup_{T \to \infty} \int_0^{\delta} \rho_0^T(u) \, du = 0
\end{equation*}
Furthermore, since the maximum jump size of $\wb \Lambda^T$, that is $(1-a_0^T)(T\nu^T)^{-1}$, goes to zero, we conclude that $\wb \Lambda^T$ is C-tight using Proposition VI-3.26 in \cite{jacod1987limit}. The rest of the proof can be found in \cite{jaisson2016rough}.

\subsection{Proof of Proposition \ref{prop:holder_fundamental}} Suppose $\alpha_0 < 1/2$
with the convention $ f^{\alpha_0,\lambda_0}(u)=0$ for $u\le 0$.
We define the forward increment operator $\DeltaH f(t) := f(t+h)-f(t)$ for $t,h>0$ and $(X,Z)$ to denote either $(F^+, Z^+)$ or $(F^-,Z^-)$. We set
\begin{equation*}
V(t,h) := \E\big[(X_{t+h}-X_t)^2\big]. \\
\end{equation*}	
Proving Proposition \ref{prop:holder_fundamental} is equivalent to proving that
\begin{equation*}
V(t,h) = O(h^{4\alpha_0}).
\end{equation*}
We first decompose $X_{t} = g(t) + \widehat X_{t}$ where
\begin{equation*}
g(t)   := \E[X_t] = \int_0^t s\, f^{\alpha_0,\lambda_0}(t-s)\,\dd s,
\end{equation*}
and 
\begin{equation*}
\widehat X_t = \int_0^t  f^{\alpha_0,\lambda_0}(t-s) Z_s\,\dd s.
\end{equation*}

With these notations, we obtain
\begin{align*}
V(t,h) &= \E\Big[\big(\DeltaH g(t) + \DeltaH \widehat X_t\big)^2\Big] 
= \big(\DeltaH g(t)\big)^2 + \E\big[(\DeltaH \widehat X_t)^2\big]
\end{align*}
since $\E[\Delta_h \widehat X_t] = \Delta_h \E[\widehat X_t] = 0$. Note that 
\begin{align*}
\DeltaH \widehat X_t
&=\int_0^{t+h}  f^{\alpha_0,\lambda_0}(t+h-s) Z_s\,\dd s -\int_0^t  f^{\alpha_0,\lambda_0}(t-s) Z_s\,\dd s  \\
&= \int_0^{t+h} \big( f^{\alpha_0,\lambda_0}(t+h-s)- f^{\alpha_0,\lambda_0}(t-s)\big) Z_s\,\dd s \\
&= \int_0^{t+h} \DeltaH  f^{\alpha_0,\lambda_0}(t-s)\, Z_s\,\dd s.
\end{align*}
Thus, we write
\begin{align*}
\E\big[(\DeltaH \widehat X_t)^2\big]
&= \int_0^{t+h}\!\int_0^{t+h} \DeltaH f^{\alpha_0,\lambda_0}(t-s)\,\DeltaH f^{\alpha_0,\lambda_0}(t-v)\,\E[Z_s Z_v] \,\dd s\dd v \\
&= 2\int_0^{t+h} g(s)\,\DeltaH f^{\alpha_0,\lambda_0}(t-s)\left(\int_s^{t+h} \DeltaH f^{\alpha_0,\lambda_0}(t-v)\,\dd v\right)\dd s.
\end{align*}
We introduce
\begin{equation*}
\varrho^{\alpha_0,\lambda_0}(x) := \int_0^x  f^{\alpha_0,\lambda_0}(y)\,\dd y, \qquad x\ge 0.
\end{equation*}
so that we have  $\int_s^{t+h}\!\DeltaH f^{\alpha_0,\lambda_0}(t-v)\,\dd v=\DeltaH \varrho^{\alpha_0,\lambda_0}(t-s)$, and thus
\begin{equation*}
\E\big[(\DeltaH \widehat X_t)^2\big] 
= 2\int_0^{t+h} g(s)\,\DeltaH f^{\alpha_0,\lambda_0}(t-s)\,\DeltaH \varrho^{\alpha_0,\lambda_0}(t-s)\,\dd s.
\end{equation*}
Hence,
\begin{equation}
\label{eq:variance}
\begin{split}
V(t,h)
&= \big(\DeltaH g(t)\big)^2 + 2\int_0^{t+h} g(s)\,\DeltaH f^{\alpha_0,\lambda_0}(t-s)\,\DeltaH \varrho^{\alpha_0,\lambda_0}(t-s)\,\dd s \\
&= \big(\DeltaH g(t)\big)^2 + 2\int_0^{t} g(t-s)\,\DeltaH f^{\alpha_0,\lambda_0}(s)\,\DeltaH \varrho^{\alpha_0,\lambda_0}(s)\,\dd s \\
&\phantom{=}\; + 2\int_t^{t+h} g(s)\,\DeltaH f^{\alpha_0,\lambda_0}(t-s)\,\DeltaH \varrho^{\alpha_0,\lambda_0}(t-s)\,\dd s.
\end{split}
\end{equation}

We would like to bound $g$.
We have for $0\le t\le 1$ and $0<h\le 1-t$,
\begin{equation*}
|g(t)| = \Big|\int_0^t sf^{\alpha_0, \lambda_0}(t-s)ds\Big| \le |\varrho^{\alpha_0, \lambda_0}(t)| \leq 1,
\end{equation*}
and
\begin{align*}
|\DeltaH g(t)| & = \Big|\int_0^{t+h}sf^{\alpha_0, \lambda_0}(t+h-s)ds - \int_0^{t}sf^{\alpha_0, \lambda_0}(t-s)ds\Big|\\
& = \Big|\int_0^{t+h} (t+h-s) f^{\alpha_0, \lambda_0}(s)ds - \int_0^{t}  (t-s)f^{\alpha_0, \lambda_0}(s)ds\Big| \\
& \leq h \Big|\int_0^t f^{\alpha_0, \lambda_0}(s)ds\Big| + \Big| \int_t^{t+h} (t+h-s)f^{\alpha_0, \lambda_0}(s)ds\Big| \\
& \leq h\varrho^{\alpha_0, \lambda_0}(t) + h |\varrho^{\alpha_0, \lambda_0}(t+h) - \varrho^{\alpha_0, \lambda_0}(t)|\\
& \leq h.
\end{align*}
In particular,
\begin{equation*}
\big(\DeltaH g(t)\big)^2 = O(h^2). 
\end{equation*}

Furthermore, for $s\in[t,t+h]$, \begin{equation*}\DeltaH f^{\alpha_0,\lambda_0}(t-s)= f^{\alpha_0,\lambda_0}(t+h-s)\qquad  \text{and}  \qquad
\DeltaH \varrho^{\alpha_0, \lambda_0}(t-s)=\varrho^{\alpha_0, \lambda_0}(t+h-s).\end{equation*} By the mean-value theorem, for each $t \leq s$, we can write
$g(s) = g(t) + g'(\xi_t(s))(s-t)$
for some $t \leq \xi_t(s)\leq s$.
Therefore, the last term of \eqref{eq:variance} becomes
\begin{align*}
\int_t^{t+h} g(s)\,\DeltaH f^{\alpha_0,\lambda_0}(t-s)\,& \DeltaH \varrho^{\alpha_0, \lambda_0}(t-s)\,\dd s
= g(t) \int_t^{t+h}  f^{\alpha_0,\lambda_0}(t+h-s)\varrho^{\alpha_0, \lambda_0}(t+h-s)\,\dd s \\
& + \int_t^{t+h} g'(\xi_t(s))(s-t)\, f^{\alpha_0,\lambda_0}(t+h-s)\varrho^{\alpha_0, \lambda_0}(t+h-s)\,\dd s.
\end{align*}
A change of variables gives
\begin{align*}
\int_t^{t+h}  f^{\alpha_0,\lambda_0}(t+h-s)\varrho^{\alpha_0, \lambda_0}(t+h-s)\,\dd s
&= \int_0^{h}  f^{\alpha_0,\lambda_0}(v)\varrho^{\alpha_0, \lambda_0}(v)\,\dd v 
= \tfrac12 \varrho^{\alpha_0, \lambda_0}(h)^2. \label{eq:betaF-half-square}
\end{align*}
Because $\Delta_h g$ is bounded, $g'$ is also bounded on $[0,1]$. We obtain
\begin{align*}
\Big|\int_t^{t+h} g'(\xi_t(s))(s-t)\, f^{\alpha_0,\lambda_0}(t+h-s)&\varrho^{\alpha_0, \lambda_0}(t+h-s)\,\dd s \Big| \\
& \le C' h \int_t^{t+h}  f^{\alpha_0,\lambda_0}(t+h-s)\varrho^{\alpha_0, \lambda_0}(t+h-s)\,\dd s \\
&= \frac{C' h}{2} \varrho^{\alpha_0, \lambda_0}(h)^2
\end{align*}
for some constant $C' > 0$. Consequently,
\begin{equation*}
\int_t^{t+h} g(s)\,\DeltaH f^{\alpha_0,\lambda_0}(t-s)\,\DeltaH \varrho^{\alpha_0, \lambda_0}(t-s)\,\dd s
= \frac{1}{2} \varrho^{\alpha_0, \lambda_0}(h)^2\,\big(1+O(h)\big). \label{eq:last-term}
\end{equation*}
We now make explicit the behavior of $\varrho^{\alpha_0,\lambda_0}(h)$ as $h\to0$.
\begin{lemma}
As $h$ goes to $0$, we have
    \begin{equation*}
\varrho^{\alpha_0,\lambda_0}(h) = \frac{\lambda_0}{\Gamma(\alpha_0+1)} h^{\alpha_0} + O(h^{2\alpha_0}).
\end{equation*}
\end{lemma}
\begin{proof}
We recall that the Mittag--Leffler density satisfies
\begin{equation*}
f^{\alpha_0,\lambda_0}(t) = \lambda_0 t^{\alpha_0-1} E_{\alpha_0,\alpha_0}(-\lambda_0 t^{\alpha_0}), 
\qquad 
\varrho^{\alpha_0,\lambda_0}(t) = \lambda_0 t^{\alpha_0} E_{\alpha_0,\alpha_0+1}(-\lambda_0 t^{\alpha_0}).
\end{equation*}

Using the series expansion of the Mittag--Leffler function
\begin{equation*}
E_{\alpha_0,\beta}(x) = \sum_{n=0}^\infty \frac{x^n}{\Gamma(\alpha_0 n + \beta)},
\end{equation*}
we obtain
\begin{equation*}
\varrho^{\alpha_0,\lambda_0}(t) 
= \lambda_0 t^{\alpha_0} \sum_{n=0}^\infty \frac{(-\lambda_0 t^{\alpha_0})^n}{\Gamma(\alpha_0 n + \alpha_0 + 1)}
= \sum_{n=1}^\infty \frac{(-\lambda_0 t^{\alpha_0})^n}{\Gamma(\alpha_0 n + 1)} = 1 - E_{\alpha_0,1}(-\lambda_0 t^{\alpha_0}).
\end{equation*}
From the power series expansion of $E_{\alpha_0,1}$ around $0$,
\begin{equation*}
E_{\alpha_0,1}(-\lambda_0 h^{\alpha_0}) 
= 1 - \frac{\lambda_0}{\Gamma(\alpha_0+1)} h^{\alpha_0} + \frac{\lambda_0^2}{\Gamma(2\alpha_0+1)} h^{2\alpha_0} + O(h^{3\alpha_0}),
\end{equation*}
we obtain
\begin{equation*}
\varrho^{\alpha_0,\lambda_0}(h) = \frac{\lambda_0}{\Gamma(\alpha_0+1)} h^{\alpha_0} + O(h^{2\alpha_0}), 
\qquad h \to 0.
\end{equation*}

\end{proof}

We are interested now in the second term of \eqref{eq:variance}. Let
\begin{equation*}
H(t) := E_{\alpha_0,1}(-\lambda_0 t^{\alpha_0}).
\end{equation*}
Then $H\in C^1(\mathbb{R}_+)$, $H$ is continuous on $(0,\infty)$ and decreasing, and
\begin{equation*}
f^{\alpha_0,\lambda_0}(t) = -H'(t), 
\qquad 
\varrho^{\alpha_0,\lambda_0}(t) = 1 - H(t).
\end{equation*}
Consider
\begin{equation*}
I(t,h):= \int_0^{t} g(t-s)\,\DeltaH f^{\alpha_0,\lambda_0}(s)\,\DeltaH \varrho^{\alpha_0,\lambda_0}(s)\,\dd s
= \int_0^t g(t-s)\,(\DeltaH H') (s)\,(\DeltaH H)(s)\,\dd s,
\end{equation*}
Integrating by parts, we have
\begin{align*}
I(t,h)
&= \left[ g(t-s)\,\frac{(\DeltaH H(s))^2}{2} \right]_{s=0}^{s=t}
 - \frac12 \int_0^t g'(t-s)\,(\DeltaH H(s))^2\,\dd s \\
&= -\frac12 g(t)\,(\DeltaH H(0))^2 
 - \frac12 \int_0^t \varrho^{\alpha_0,\lambda_0}(t-s)\,(\DeltaH H(s))^2\,\dd s, \label{eq:I1-split}
\end{align*}
where we used $g(0)=0$ and $g'=\varrho^{\alpha_0,\lambda_0}$. Since $(\DeltaH H(0))^2=(\DeltaH \varrho^{\alpha_0,\lambda_0}(0))^2=\varrho^{\alpha_0,\lambda_0}(h)^2$, the first term equals $-\tfrac12 g(t)\varrho^{\alpha_0,\lambda_0}(h)^2$. Moreover, for all $\delta > 0$, we have
\begin{align*}
\int_0^t \varrho^{\alpha_0,\lambda_0}(t-s)\,(\DeltaH H(s))^2\,\dd s 
&\le \varrho^{\alpha_0,\lambda_0}(t) \int_0^t (\DeltaH H(s))^2\,\dd s \\
& \le \varrho^{\alpha_0,\lambda_0}(t)\left\{\int_0^{\delta} (\DeltaH H(s))^2\,\dd s + \int_{\delta}^{t} (\DeltaH H(s))^2\,\dd s\right\}.
\end{align*}
Using that $H$ is bounded and for $u,h > 0$ \begin{equation*} (\DeltaH H(s))^2 = \Big(\int_s^{s+h}H'(v)dv\Big)^2 =\Big(\int_s^{s+h}f^{\alpha_0,\lambda_0}(v)dv\Big)^2 \leq h^2 (f^{\alpha_0,\lambda_0}(s))^2. \end{equation*}
Moreoever, since $\alpha_0 < 1/2$, $E_{\alpha_0,\alpha_0}(-s) \leq 1$ for any positive $s$, and we have that \begin{equation*}
(f^{\alpha_0,\lambda_0}(s))^2 \le \lambda_0^2s^{2\alpha_0-2}.
\end{equation*}
Therefore, we write
\begin{align*}
\int_0^t \varrho^{\alpha_0,\lambda_0}(t-u)\,(\DeltaH H(u))^2\,\dd u
&\le \varrho^{\alpha_0,\lambda_0}(t)\left\{C\,\delta + h^2\int_{\delta}^{t} ( f^{\alpha_0,\lambda_0}(u))^2\,\dd u\right\} \\
&\le \varrho^{\alpha_0,\lambda_0}(t)\left\{C\,\delta + h^2\int_{\delta}^{\infty} ( f^{\alpha_0,\lambda_0}(u))^2\,\dd u\right\} \\
&\le \varrho^{\alpha_0,\lambda_0}(t)\left\{C\,\delta + h^2\int_{\delta}^{\infty} \lambda_0^2 u^{2\alpha_0-2}\,\dd u\right\} \\
&= \varrho^{\alpha_0,\lambda_0}(t)\left\{C\,\delta + h^2\,\lambda_0^2\,\frac{\delta^{2\alpha_0-1}}{1-2\alpha_0}\right\},
\end{align*}
which is finite since $\alpha_0<1/2$. Choosing $\delta=h^{1/(1-\alpha_0)}$ balances the two terms, giving
\begin{equation*}
\int_0^t \varrho^{\alpha_0,\lambda_0}(t-u)\,(\DeltaH H(u))^2\,\dd u \le C' \varrho^{\alpha_0,\lambda_0}(t)\, h^{1/(1-\alpha_0)}.
\end{equation*}
Using $\alpha_0(1-\alpha_0)\le 1/4$ we have $\tfrac{1}{1-\alpha_0}\ge 4\alpha_0$, hence for $h\in(0,1]$,
\begin{equation*}
h^{1/(1-\alpha_0)} \le h^{4\alpha_0}.
\end{equation*}
Therefore,
\begin{equation*}
I(t,h) = \frac12 g(t)\varrho^{\alpha_0,\lambda_0}(h)^2 + O\big(\varrho^{\alpha_0,\lambda_0}(t)\,h^{4\alpha_0}\big). \label{eq:I1-final}
\end{equation*}
Using Lemma \eqref{eq:I1-final}, we obtain
\begin{equation*}
I(t,h) = \frac{g(t)\,\lambda_0^2}{2\,\Gamma(\alpha_0+1)^2}\, h^{2\alpha_0} + O\big(h^{4\alpha_0\wedge1}\big).
\end{equation*}

Hence, going back to \eqref{eq:variance}, we write
\begin{align*}
V(t,h) 
&= \big(\DeltaH g(t)\big)^2 
+ 2 I(t,h) + 2 \int_t^{t+h} g(u)\,\DeltaH f^{\alpha_0,\lambda_0}(t-u)\,\DeltaH \varrho^{\alpha_0,\lambda_0}(t-u)\,\dd u \\
&= O(h^2) + 2\cdot \frac{g(t)\,\lambda_0^2}{2\,\Gamma(\alpha_0+1)^2}\, h^{2\alpha_0} + 2\cdot \frac{1}{2} \varrho^{\alpha_0, \lambda_0}(h)^2\,\big(1+O(h)\big) + O\big(h^{4\alpha_0}\big) \\
&= \frac{2\lambda_0^2}{\Gamma(\alpha_0+1)^2}\, (1+g(t))\, h^{2\alpha_0} + O\big(h^{4\alpha_0 \wedge 1}\big) + O(h^2).
\end{align*}
Since $2\alpha_0<1$, the remainder $O(h^2)$ is negligible with respect to $h^{2\alpha_0}$ as $h$ goes to $0$.\\
In summary, for $\alpha_0\in(0,\tfrac12)$ the following holds uniformly for $t,h>0$, as $h$ tends to $0$,
\begin{equation*}
V(t,h) = \E\big[(X_{t+h}-X_t)^2\big] = \frac{2\lambda_0^2}{\Gamma(\alpha_0+1)^2}\, (1+g(t))\, h^{2\alpha_0} + O\big(h^{4\alpha_0 \wedge 1}\big).
\end{equation*}
We conclude that $F^+$ and $F^-$ are exactly $2\alpha_0$--Hölder continuous in $L^2$.
\end{proof}

\subsection{Proof of Proposition \ref{prop:cnvg_fundamental}}

From Theorem \ref{thm:cnvg_fundamental} we can see that 
\begin{equation*}F_{t} = F^{+}_{t} + F^{-}_{t} = 2\int_0^t s\,f^{\alpha_0,\lambda_0}(t - s)\,ds
+
\frac{1}{\sqrt{\mu_0\,\lambda_0}}
\int_0^t f^{\alpha_0,\lambda_0}(t - s)\,Z^{F}_{s}ds\end{equation*}
where  $Z^{F} =  Z^+ + Z^-$. Notice that $Z^{F}$ is a continuous martingale with quadratic variation $F$. 
On the other hand, the process $V$ satisfies
\begin{equation*}
V_t = F_t^+ - F_t^- = \frac{1}{\sqrt{\mu_0\,\lambda_0}}
\int_0^t f^{\alpha_0,\lambda_0}(t - s)\,Z^{V}_{s}ds
\end{equation*}
where  $Z^{V} =   Z^+ - Z^-$. Note also that  $Z^{V}$ is a continuous martingale with quadratic variation $F$. \\
We can compute the quadratic covariance of the two resulting martingales
\begin{align*}
<\!Z^V, Z^F\!> \, = \, <\!Z^+ - Z^-, Z^+ + Z^-\!> \, = F^+ - F^- = V.\\
\end{align*}

\subsection{Proof of Theorem \ref{thm:cnvg_reaction}}

The proof relies on replicating the findings of \cite{eleuch2018microstructural} with the stochastic time-varying baseline $\bfmu^{T}$. We start by providing multiple elements needed for the proof.\\
First, note that we have
\begin{align*}
\int_0^t \bflambda_s^{T} \, ds
=
\int_0^t \bfmu_s^{T} \, ds
+
\int_0^t \psi^T(t-s) \cdot \int_0^s \bfmu_u^{T} \, du \, ds
+
\int_0^t \psi^T(t-s) \cdot \mathbf{M}^T_s \, ds.
\end{align*}
Now, $\bfmu^T = 
\phi^T * 
d
\mathbf{F}^T$
and since $\mathbf{F}^T_0 = \mathbf{0}$ we have
\begin{equation*}
\int_0^t \bfmu_s^{T} \, ds
=
\phi^T * 
\mathbf{F}^T_t.
\end{equation*}
Using also the identity $\psi^T * \phi^T = \psi^T - \phi^T$,
we obtain
\begin{align*}
\int_0^t \bflambda_s^{T} \, ds
=
\int_0^t \psi^T(t-s)\cdot \mathbf{F}^T_s \, ds
+
\int_0^t \psi^T(t-s)\cdot \mathbf{M}^T_s \, ds.
\end{align*}
In this setting, it is more suitable to work with the two-dimensional rescaled processes
\begin{equation*}
\begin{split}
    &\wb {\mathbf{N}}_t^{T} = \frac{(1-a^T_0)(1-a^T_1)}{T\nu^T} \mathbf{N}_{tT}^{T},
    \\
    &\wb \bfLambda_t^{T} = \frac{(1-a^T_0)(1-a^T_1)}{T\nu^T} \bfLambda_{tT}^{T},
    \\
    &\wb {\mathbf{M}}_t^{T} = \Big(\frac{(1-a^T_0)(1-a^T_1)}{T\nu^T}\Big)^{1/2} \mathbf{M}_{tT}^{T}.
\end{split}
\end{equation*}
The scaled unsigned reaction flow is then given by \begin{equation*}\transp{v_1} \cdot\wb {\mathbf{N}}^{T} = \wb {N}^{+,T} + \wb {N}^{-,T},\end{equation*} and the scaled signed reaction flow by	\begin{equation*}\transp{v_2} \cdot\wb {\mathbf{N}}^{T} = \wb {N}^{+,T} - \wb {N}^{-,T}.\end{equation*}
We can then write
\begin{align*}
\wb \bfLambda_t^{T} &= \frac{(1-a^T_0)(1-a^T_1)}{T\nu^T}
\int_0^{tT} \bflambda_s^{T} \, ds\\
& =
\int_0^t T(1 - a_1^T) \psi^T(T(t-s)) \cdot \wb {\mathbf{F}}^T_{s} \, ds + 
\frac{1-a^T_0}{T\nu^T} \int_0^t T(1 - a_1^T) \psi^T(T(t-s))\cdot  \mathbf{M}^T_{sT} \, ds.
\end{align*}
Note that
\begin{equation*}
\E[\wb \bfLambda_t^{T} ] = \int_0^t T(1 - a_1^T) \psi^T(T(t-s)) \cdot \E[\wb {\mathbf{F}}^T_{s}] \, ds,
\end{equation*}
and from Section \ref{proof:cnvg_fundamental}, we know that $\E[\wb {\mathbf{F}}^{\pm, T}_{s}] \leq 1$, then 
\begin{equation*}
\transp{v_1} \cdot \E[\wb \bfLambda_t^{T} ] \leq T (1 - a_1^T)  \transp{v_1}\cdot \Big( \int_0^t  \psi^T(T(t-s)) ds\Big)\cdot v_1 \leq (1 - a_1^T) \varrho\left(\int_0^{\infty} \psi^T(s) d s\right) < 1.
\end{equation*}
Therefore, using the Burkholder-Davis-Gundy inequality, we get that 
\begin{equation*}
\E\Big[\sup_{t \leq 1} \bignorm{\wb {\mathbf{M}}_t^{T}}_2^2\Big] \leq C
\end{equation*} 
for some constant $C > 0$. \\

For $i = 1, 2$, $v_i$ is the eigenvector associated with the eigenvalue $k_i$ so we have
\begin{equation*}
\phi^T \cdot v_i = k_i^T v_i.
\end{equation*}

By induction,
\begin{equation*}
v_i^T \cdot (\phi^T)^{*n} = (k_i^T)^{*n}\,v_i^T,
\end{equation*}
and we define scalar kernels
\begin{equation*}
\psi_i^T(x)=\sum_{n\ge1}(a_1^T)^n(k_i^T)^{*n}(x),
\quad
\rho_i^T(x)=T(1-a_1^T)\psi_i^T(Tx),
\quad
\varrho_i^T(t)=\int_0^t\rho_i^T(s)\,ds.
\end{equation*}
Consequently, we have $\transp{v_i} \cdot\psi^T = \psi_i^T \transp{v_i}$ and
\begin{align*}
\transp{v_i} \cdot \wb \bfLambda_t^{T} 	=
\int_0^t \rho_i^T(t-s) \transp{v_i} \cdot \wb {\mathbf{F}}^T_{s} \, ds
+
c^T \int_0^t \rho_i^T(t-s)  \transp{v_i} \cdot \wb {\mathbf{M}}^T_{s} \, ds
\end{align*}
where \begin{equation*}c^T = \sqrt{(1-a_0^T)/(T\nu^T(1-a_1^T))} \to  \sqrt{\frac{1}{\lambda_1 \mu_1}}.\end{equation*}
We are interested in studying the convergence of this process for $i \in \{1,2\}$. \\

\textbf{Convergence of $\transp{v_i} \cdot \wb{ \mathbf{N}}^{T}_t  - \transp{v_i} \cdot \wb{ \bfLambda}^{T}_t $}. We have
\begin{equation*}
\transp{v_i} \cdot \wb{ \mathbf{N}}^{T}_t  - \transp{v_i} \cdot \wb{ \bfLambda}^{T}_t = \Big(\frac{(1-a^T_0)(1-a^T_1)}{T\nu^T}\Big)^{1/2}\transp{v_i} \cdot \wb {\mathbf{M}}^T_{t}
\end{equation*}
Using Doob's inequality, we get
\begin{align*}
\E\Big[ \sup_{t \leq 1} \Big |\transp{v_i} \cdot \wb{ \mathbf{N}}^{T}_t  - \transp{v_i} \cdot \wb{ \bfLambda}^{T}_t \Big|^2   \Big] & \leq 4 \Big(\frac{(1-a^T_0)(1-a^T_1)}{T\nu^T}\Big) \E[|\transp{v_i} \cdot  \wb {\mathbf{M}}_T^{T}|^2]\\
& \leq 4 \Big(\frac{(1-a^T_0)(1-a^T_1)}{T\nu^T}\Big) \bignorm{v_i}^2 \E[\bignorm{ \wb {\mathbf{M}}_T^{T}}_2^2]\\
& \leq C^{'} \Big(\frac{(1-a^T_0)(1-a^T_1)}{T\nu^T}\Big).
\end{align*}
Since $(1-a^T_0)(1-a^T_1)(T\nu^T)^{-1}$ is of the order of $T^{-2\alpha_1}$, we obtain the convergence to zero in $L^2$ and in probability of $\transp{v_i} \cdot \wb{ \mathbf{N}}^{T}_t  - \transp{v_i} \cdot \wb{ \bfLambda}^{T}_t $.\\

\textbf{Convergence of $\transp{v_2} \cdot \wb{ \mathbf{N}}^{T}_t $}. We know from \cite{eleuch2018microstructural} that $\varrho_2^T$ converges uniformly to zero and we write
\begin{align*}
\E\Big[\Big|\int_0^t \rho_2^T(t-s) \transp{v_2} \cdot \wb {\mathbf{F}}^{T}_s \, ds \Big|\Big] \leq \bignorm{v_2} \sup_{t \leq 1} \E[\bignorm{\wb {\mathbf{F}}_t^{T}}] \varrho_2(t).\\
\end{align*}
Using the fact that $\E[\bignorm{\wb {\mathbf{F}}_t^{T}}]$ is bounded, which has been proven in Section \ref{proof:cnvg_fundamental}, we conclude the first integral of $\transp{v_2} \cdot \wb{ \mathbf{N}}^{T}_t$ converges to zero in $L^1$. For the second integral, we write using an integration by parts
\begin{align*}
c^T \int_0^t \rho_2^T(t-s)  \transp{v_2}  \cdot \wb {\mathbf{M}}^T_{s} \, ds = 	c^T \int_0^t \varrho_2^T(t-s)(\transp{v_2} \cdot \wb {\mathbf{M}}^T_{s}).
\end{align*}
Thus, there exists a constant $C^{'}$ such that
\begin{align*}
\E\Big[\Big(c^T \int_0^t \rho_2^T(t-s)  \transp{v_2}  \cdot \wb {\mathbf{M}}^T_{s} \, ds\Big)^2\Big] \leq C^{'} \int_0^t (\varrho_2^T(s))^2ds
\end{align*}
and therefore the second integral converges to 0 in $L^2$. Therefore we obtain that $\transp{v_2} \cdot \wb \bfLambda_t^{T}$ goes to zero in $L^1$. Consequently, $\transp{v_2} \cdot \wb {\mathbf{N}^{T}}$ converges to 0 in probability and in $L^1$.\\

\textbf{Convergence of $\transp{v_1} \cdot \wb{ \mathbf{N}}^{T}_t $}. We write
\begin{align*}
\transp{v_1} \cdot \wb \bfLambda_t^{T} = 
\int_0^t \rho_1^T(t-s) \transp{v_1} \cdot \wb {\mathbf{F}}^T_{s} \, ds
+
c^T \int_0^t \rho_1^T(t-s)  \transp{v_1} \cdot  \wb {\mathbf{M}}^T_{s} \, ds
\end{align*}
Using the same arguments and methodology as in Section \ref{proof:cnvg_fundamental}, we get that $\transp{v_1} \cdot \wb \bfLambda^T$ is C-tight and we conclude that $(\transp{v_1} \cdot  \wb {\mathbf{N}}^T, \transp{v_1} \cdot  \wb \bfLambda^T, \transp{v_1} \cdot  \wb {\mathbf{M}}^T)$ is C-tight. Furthermore, if $(X, X, Z)$ is a limit point of $(\transp{v_1} \cdot  \wb {\mathbf{N}}^T, \transp{v_1} \cdot  \wb \bfLambda^T,\transp{v_1} \cdot  \wb {\mathbf{M}}^T)$, then $Z$ is a continuous martingale with $[Z,Z]=X$.\\

Moreover, we know from \cite{jaisson2016rough} that the sequence of measures with density $\rho_1^T(x)$ converges weakly towards the measure with density $\lambda_1 x^{\alpha_1-1} E_{\alpha_1, \alpha_1}\left(-\lambda_1 x^\alpha_1\right)$. In particular, over $[0,1]$,
\begin{equation*}
\varrho_1^T(t)=\int_0^t \rho_1^T(x) d x
\end{equation*}
converges uniformly towards
\begin{equation*}
\varrho^{\alpha_1, \lambda_1}(t)=\int_0^t f^{\alpha_1, \lambda_1}(x) d x.
\end{equation*}
Therefore, using the same approach as in \cite{jaisson2016rough} yields
\begin{equation*}
\begin{split}
    &\int_0^t \rho_1^T(t-s) \transp{v_1} \cdot \wb {\mathbf{F}}^T_{s} \, ds \to  \int_0^t f^{\alpha_1, \lambda_1}(t-s) F_{s} \, ds
    \\
    &c^T \int_0^t \rho_i^T(t-s)  \transp{v_i} \cdot  \wb {\mathbf{M}}^T_{s} \, ds \to \sqrt{\frac{1}{\lambda_1 \mu_1}} \int_0^t f^{\alpha_1, \lambda_1}(t-s) Z_{s} \, ds
\end{split}
\end{equation*}
where $F$ is the scaling limit of $\transp{v_1} \cdot\wb {\mathbf{F}}^T$ from Proposition \ref{prop:cnvg_fundamental}.\\

\textbf{Convergence of $(\wb N^{+, T}, \wb N^{-, T}, \wb \Lambda^{+, T}, \wb  \Lambda^{-, T}, \wb M^{+, T}, \wb M^{-, T})$}. We use the fact that the sum process $(\transp{v_1} . \wb {\mathbf{N}}^T, \transp{v_1}  .\wb \bfLambda^T,\transp{v_1} . \wb {\mathbf{M}}^T)$ is C-tight, which implies the C-tightness of the process \\
$(\wb N^{+, T}, \wb N^{-, T}, \wb \Lambda^{+, T}, \wb  \Lambda^{-, T}, \wb M^{+, T}, \wb M^{-, T})$. Furthermore, using the same arguments as in Section \ref{appendix:proofs}, the previous result, and the fact that
\begin{equation*}
\wb N^{+, T} = \frac{1}{2}(\transp{v_1} . \wb {\mathbf{N}}^T + \transp{v_2} . \wb {\mathbf{N}}^T) \quad \text{and} \quad \wb N^{-, T} = \frac{1}{2}(\transp{v_1} . \wb {\mathbf{N}}^T - \transp{v_2} . \wb {\mathbf{N}}^T),
\end{equation*}
if $(X,X,X,X,Z^{+},Z^{-})$ is an accumulation point of $(\wb N^{+, T}, \wb N^{-, T}, \wb \Lambda^{+, T}, \wb  \Lambda^{-, T}, \wb M^{+, T}, \wb M^{-, T})$, then 
\begin{align*}
X_t &=  \frac{1}{2}\int_0^t f^{\alpha_1, \lambda_1}(t-s) F_{s} \, ds
+
\frac{1}{2\sqrt{\lambda_1 \mu_1}}  \int_0^t f^{\alpha_1, \lambda_1}(t-s)Z_s \, ds \qquad \text{and} \qquad
Z_t = Z^+_t + Z^-_t 
\end{align*}
where $Z^+$ and $Z^-$ are two continuous martingales with quadratic variation $X$ and zero quadratic covariation.
Seeing that $F$ is smoother than $Z$, the regularity of $X$ is determined by the second integral, which is $(H_1-\varepsilon)$-Holder continuous for every $\varepsilon > 0$ on $[0,1]$, with $H_1 = \alpha_1 - 1/2$.

\subsection{Proof of Hölder regularity in Theorem \ref{thm:cnvg_reaction}}
From the previous section, we know that:
\begin{itemize}
    \item $X$ is Lipschitz continuous,
    \item $Z$ is $(1/2 - \varepsilon)$-Hölder continuous for all $\varepsilon >0$, since its quadratic variation, which is $X$, is continuous,
    \item $F$ is $(2\alpha_0 - \varepsilon)$-Hölder continuous for all $\varepsilon >0$,
    \item $Z^F$ is $(\alpha_0 - \varepsilon)$-Hölder continuous for all $\varepsilon >0$.
\end{itemize}
Then for $0 < \gamma < 1$, we know from Proposition A.1 in \cite{jaisson2016rough} that
\begin{itemize}
    \item $X$ admits a fractional derivative of order $\gamma$ and $D^\gamma X$ is $(1 - \gamma)$-Hölder regular,
    \item If $\gamma < 2\alpha_0 = \alpha_1$, then $F$ admits a fractional derivative of order $\gamma$ and $D^\gamma F    $ is $(2\alpha_0 - \gamma - \varepsilon)$-Hölder regular for all $\varepsilon > 0$,
    \item If $\gamma < 1/2$, then $Z$ admits a fractional derivative of order $\gamma$ and $D^\gamma Z    $ is $(1/2 - \gamma - \varepsilon)$-Hölder regular for all $\varepsilon > 0$.
\end{itemize}
Let $1/2 < \gamma < \alpha_1$. From Proposition 3.1 and Corollary A.2 in  \cite{jaisson2016rough}, we have
\begin{align*}
    X_t &= \frac{1}{2}\int_0^t f^{\alpha_1, \lambda_1}(t-s) F_{s} \, ds
+
\frac{1}{2\sqrt{\lambda_1 \mu_1}}  \int_0^t f^{\alpha_1, \lambda_1}(t-s)Z_s \, ds \\
&= \frac{1}{2}\int_0^t D^\gamma f^{\alpha_1, \lambda_1}(t-s) I^\gamma F_{s} \, ds
+
\frac{1}{2\sqrt{\lambda_1 \mu_1}}  \int_0^t D^\gamma f^{\alpha_1, \lambda_1}(t-s) I^\gamma Z_s \, ds
\end{align*}
Furthermore, $F$ and $Z$ are fractionally differentiable and we have
\begin{equation*}
I^\gamma F_s = \int_0^s D^{1- \gamma} F_u du
\quad\quad\text{ and }\quad\quad I^\gamma Z_s = \int_0^s D^{1- \gamma} Z_u du.
\end{equation*}

We rewrite the expression of $X$ as follows
\begin{equation}
\label{eq:X}
    X_t = \frac{1}{2} \int_0^t \int_0^s D^\gamma f^{\alpha_1, \lambda_1}(t-s) D^{1-\gamma} F_u du ds
    + \frac{1}{2\sqrt{\lambda_1 \mu_1}}  \int_0^t \int_0^s D^\gamma f^{\alpha_1, \lambda_1}(t-s) D^{1-\gamma} Z_u du ds
\end{equation}

We use Fubini's theorem and we write
\begin{align*} 
\int_0^t \int_0^s D^\gamma f^{\alpha_1, \lambda_1}(t-s) D^{1-\gamma} Z_u d u d s  & =\int_0^t \int_u^t D^\gamma f^{\alpha_1, \lambda_1}(t-s) D^{1-\gamma} Z_u d s d u \\ & =\int_0^t \int_u^t D^\gamma f^{\alpha_1, \lambda_1}(s-u) D^{1-\gamma} Z_u d s d u \\ & =\int_0^t \int_0^s D^\gamma f^{\alpha_1, \lambda_1}(s-u) D^{1-\gamma} Z_u d u d s.
\end{align*}
Applying the same computations to the first integral in \eqref{eq:X}, we get
\begin{equation*}
X_t = \int_0^t Y_s ds
\end{equation*}
with 
\begin{equation*}
Y_s = \frac{1}{2}\int_0^s D^\gamma f^{\alpha_1, \lambda_1}(s-u) D^{1-\gamma}F_u du + \frac{1}{2\sqrt{\lambda_1 \mu_1}} \int_0^s D^\gamma f^{\alpha_1, \lambda_1}(s-u) D^{1-\gamma} Z_u du.
\end{equation*}
Since $2\alpha_0 > 1/2$, we know that $F$ is smoother than $Z$, and thus the regularity of $Y$ is that of its second term. From Propositions 3.1 and A.3 in \cite{jaisson2016rough}, we have that the second integral has Hölder regularity $(\alpha_1 - \gamma)$ for $1/2 < \gamma < 1$. Thus, for every $\varepsilon > 0$, the second integral, and therefore $Y$, has Hölder regularity $(\alpha_1 - 1/2 - \varepsilon)$.

\subsection{Proof of Theorem \ref{thm:cnvg_scaled_unsigned_flow}}
First, note that $\wb F_t^{T,+} + \wb F_t^{-,T} $ scales as $T\nu^T(1-a_0^T)^{-1}$. Seeing that $(1 - a^T_1)$ is of the same order as $T^{-\alpha_1}$, we conclude that  $(1-a_0^T)(1-a_1^T)(T\nu^T)^{-1} (F_{tT}^{T,+} + F_{tT}^{-,T}) $ converges to zero.\\
Moreover, Theorem \ref{thm:cnvg_reaction} ensures that the process
\begin{equation*}\frac{(1-a_0^T)(1-a_1^T)}{T\nu^T} (N^{+,T}_{tT} + N^{-,T}_{tT}) =\wb {\mathbf{N}}^{+,T} + \wb {\mathbf{N}}^{-,T} =  \transp{v_1} \cdot\wb {\mathbf{N}}^T\end{equation*} 
is C-tight and it converges in distribution in the Skorokhod topology. Therefore, the same applies to $\wb U^T$, and if $U$ is a limit of $\transp{v_1} \cdot\wb {\mathbf{N}}^T$, then it is also a limit of $\wb U^T$ and it satisfies Equation \eqref{eq:scaling_unsigned_flow}.

\subsection{Proof of Theorem \ref{thm:cnvg_scaled_signed_flow}}
Notice that on the one hand, $(1-a_0^T)(T\nu^T)^{-1}$ is of the same order as $T^{-2\alpha_0} = T^{-\alpha_1}$. But we also know that $(1-a_1^T)$ grows like $T^{-\alpha_1}$. Therefore, we can see that
\begin{equation*}
\Big(\frac{(1-a_0^T)(1-a_1^T)}{T\nu^T}\Big)^{1/2} \qquad \text{and} \qquad \frac{1-a_0^T}{T\nu^T}
\end{equation*}
are of the same order. Thus, Theorem \ref{thm:cnvg_fundamental} guarantees that
\begin{equation*}
\Big(\frac{(1-a_0^T)(1-a_1^T)}{T\nu^T}\Big)^{1/2} \big(F_{tT}^{T,+} - F_{tT}^{-,T}\big) \to V_t
\end{equation*}
where $V$ is given by Proposition \ref{prop:cnvg_fundamental}. We just need to compute the limit of
\begin{equation*}
\Big(\frac{(1-a_0^T)(1-a_1^T)}{T\nu^T}\Big)^{1/2} \big(N_{tT}^{+,T} - N_{tT}^{-,T}\big).
\end{equation*}

We write
\begin{equation*}
N_{t}^{T,+} - N_{t}^{-,T} = M_{t}^{T,+} - M_{t}^{-,T} + \Lambda_{t}^{T,+} - \Lambda_{t}^{-,T},
\end{equation*}
and
\begin{align*}
\Lambda_{t}^{T,+} - \Lambda_{t}^{-,T}
= \transp{v_2} \bfLambda_t^{T}
&=
\int_0^t \psi_2^T(t-s) \transp{v_2} \cdot\mathbf{F}^T_s \, ds
+
\int_0^t \psi_2^T(t-s)(M^{+, T}_s - M^{-, T}_s) \, ds
\\
&=
\int_0^t \psi_2^T(t-s)  \transp{v_2} \cdot\mathbf{F}^T_s \, ds
+
\int_0^t \int_0^{t-s} \psi_2^T(u) \, du \, d(M^{+, T}_s - M^{-, T}_s)
\\
&=
\int_0^t \psi_2^T(t-s)  \transp{v_2} \cdot\mathbf{F}_s^T \, ds +
\int_0^{\infty} \psi_2^T(u) \, du (M^{+, T}_t - M^{-, T}_t)\\
& \qquad -
\int_0^t \int_{t-s}^{\infty} \psi_2^T(u) \, du \, d(M^{+, T}_s - M^{-, T}_s).
\end{align*}
Therefore we get
\begin{align*}
N_{t}^{T,+} - N_{t}^{-,T}
&=
\int_0^t \psi_2^T(t-s) \transp{v_2} \cdot\mathbf{F}^T_s \, ds
+
(1+\int_0^{\infty} \psi_2^T(u) \, du) (M^{+, T}_t - M^{-, T}_t)
\\
& \qquad -
\int_0^t \int_{t-s}^{\infty} \psi_2^T(u) \, du \, d(M^{+, T}_s - M^{-, T}_s).
\\
&=
\int_0^t \psi_2^T(t-s) \transp{v_2} \cdot\mathbf{F}^T_s\, ds
+
\frac{1}{1-a_1^T(\|\varphi_1\|_1-\|\varphi_2\|_1)} (M^{+, T}_t - M^{-, T}_t) \\
& \qquad -
\int_0^t \int_{t-s}^{\infty} \psi_2^T(u) \, du \, d(M^{+, T}_s - M^{-, T}_s).
\end{align*}
After rescaling, we obtain
\begin{equation}
\label{eq:rescaled_reaction_total_flow}
\begin{split}
\Big(\frac{(1-a_0^T)(1-a_1^T)}{T\nu^T}\Big)^{1/2} \big(N_{tT}^{+,T} - N_{tT}^{-,T}\big) &= \Big(\frac{(1-a_0^T)(1-a_1^T)}{T\nu^T}\Big)^{1/2} 	\int_0^{tT} \psi_2^T(Tt-s) \transp{v_2} \cdot\mathbf{F}^T_s\, ds \\
&+ \frac{1}{1-a_1^T(\|\varphi_1\|_1-\|\varphi_2\|_1)} (\wb M^{+, T}_t - \wb M^{-, T}_t) - R^T_t
\end{split}
\end{equation}
where 
\begin{equation*}
R^T_t = 	\int_0^t \int_{T(t-s)}^{\infty} \psi_2^T(u) \, du \, d(\wb M^{+, T}_s - \wb M^{-, T}_s).
\end{equation*}
Following the same argument as in the proof of Theorem 3.2 in \cite{eleuch2018microstructural}, we conclude the convergence of $R^T$ to zero in the sense of finite dimensional laws.\\
Furthermore, from Theorem \ref{thm:cnvg_reaction}, we know that the second term in \eqref{eq:rescaled_reaction_total_flow} converges to
\begin{equation*}
\frac{1}{1-(\|\varphi_1\|_1-\|\varphi_2\|_1)} (Z^+_t - Z^-_t).
\end{equation*}
It remains to study the first term
\begin{equation*}
\Big(\frac{(1-a_0^T)(1-a_1^T)}{T\nu^T}\Big)^{1/2} 	\int_0^{tT} \psi_2^T(Tt-s) \transp{v_2} \cdot\mathbf{F}^T_s\, ds.
\end{equation*}
After proper rescaling, we obtain
\begin{align*}
\Big(\frac{(1-a_0^T)(1-a_1^T)}{T\nu^T}\Big)^{1/2} \int_0^{tT} \psi_2^T(tT-s) \transp{v_2} \cdot F^T_s \, ds =
c^T \int_0^{t} T\psi_2^T(T(t-s)) \transp{v_2} \cdot \wb{F}^T_s  \, ds
\end{align*}
where \begin{equation*}
c^T = \sqrt{(T\nu^T(1 - a_1^T))/(1-a_0^T)} \to \sqrt{\lambda_1 \mu_1}. 
\end{equation*}
To understand its asymptotic behavior as $T$ goes to infinity, one can compute the Fourier transform $\wh {\psi_2^T(T\cdot)}$ of $\psi^T_2(T\cdot)$. We have
\begin{equation*}
\wh {\psi_2^T(T\cdot)}(z)=\int_{x \in \mathbb{R}_{+}} \psi_2^T(T x) e^{i x z} d x=\frac{1}{T} \sum_{n \geq 1} (a^T_1)^n\left(\wh k_2(z / T)\right)^n=\frac{a_1^T \wh k_2(z / T)}{T\left(1-a_1^T \wh k_2(z / T)\right)}
\end{equation*}
As $T$ goes to infinity, $\wh {k}_j(z / T)$ tends to $\norm{k_2}_1$ and recall that $\norm{k_2}_1<1$. Therefore, we see that
\begin{equation*}
T\wh {\psi_2^T(T\cdot)}(z) \to \frac{\norm{k_2}_1}{1 - \norm{k_2}_1} = \frac{\norm{\varphi_1}_1-\norm{\varphi_2}_1}{1-(\norm{\varphi_1}_1-\norm{\varphi_2}_1)}
\end{equation*}
and consequently, if we define 
\begin{equation*}
\chi_T(d t):=T \psi_2^T(T t) d t
\end{equation*} 
then we have
\begin{equation*}
\chi_T(d t) \to \frac{\norm{\varphi_1}_1-\norm{\varphi_2}_1}{1-(\norm{\varphi_1}_1-\norm{\varphi_2}_1)} \delta_0(dt)
\end{equation*}

Thus, we have shown that
\begin{equation*}
\Big(\frac{(1-a_0^T)(1-a_1^T)}{T\nu^T}\Big)^{1/2} 	\int_0^{tT} \psi_2^T(Tt-s) \transp{v_2} \cdot\mathbf{F}^T_s\, ds 
\to
 \frac{\sqrt{\lambda_1 \mu_1} (\norm{\varphi_1}_1-\norm{\varphi_2}_1)}{1-(\norm{\varphi_1}_1-\norm{\varphi_2}_1)} V_t.
\end{equation*}

Eventually, we obtain
\begin{equation*}
\Big(\frac{(1-a_0^T)(1-a_1^T)}{T\nu^T}\Big)^{1/2} S^T_{tT}
\to
\frac{\sqrt{\lambda_1 \mu_1} (\norm{\varphi_1}_1-\norm{\varphi_2}_1)}{1-(\norm{\varphi_1}_1-\norm{\varphi_2}_1)} V_t + \frac{1}{1-(\|\varphi_1\|_1-\|\varphi_2\|_1)} (Z^+_t - Z^-_t)
\end{equation*}
in the sense of finite dimensional laws.

\section{Data description and stock universe}
\label{app:data}

This appendix describes the dataset and stock universe used to estimate the Hurst exponent of signed order flow.

Our empirical analysis is conducted on a cross-section of liquid equities listed on major exchanges. The final universe consists of large-cap stocks that are actively traded throughout the sample period. Table~\ref{tab:stocks} reports the list of stocks used in the analysis, together with their tickers, exchanges, and sample years.

The data consist of trade-by-trade records obtained from BMLL, covering the period from January 2021 to December 2024.
Only regular trading days are retained; weekends, holidays, and shortened trading sessions are excluded.
All timestamps are expressed in local exchange time and restricted to standard market hours.

For each stock, the dataset contains the full sequence of executed trades with precise timestamps, traded volumes, transaction sides and prices.
The time resolution of the data is at least the millisecond (or microsecond) level, allowing for a detailed reconstruction of order flow dynamics.


The signed order flow is defined as the sequence
\begin{equation*}
    \varepsilon_t v_t,
\end{equation*}
where $\varepsilon_t \in \{+1,-1\}$ denotes the trade sign and $v_t$ the traded volume at time $t$. The unsigned order flow is defined similarly by taking $\varepsilon_t = 1$.

The signed and unsigned order flow series are subsequently aggregated over fixed time intervals to construct the increments used in the estimation of the Hurst exponents $H_0$ and $H_1$.




\begin{table}[H]
\centering
\small
\begin{tabular}{|l|l|l|l|}
\hline
\textbf{Stock Name} & \textbf{Ticker} & \textbf{Exchange} & \textbf{Region} \\
\hline
Airbus SE & AIR & Euronext Paris & Europe \\
Kering SA & KER & Euronext Paris & Europe \\
Vinci SA & DG & Euronext Paris & Europe \\
L'Or\'eal SA & OR & Euronext Paris & Europe \\
AXA SA & CS & Euronext Paris & Europe \\
Schneider Electric SE & SU & Euronext Paris & Europe \\
Cr\'edit Agricole SA & ACA & Euronext Paris & Europe \\
EssilorLuxottica SA & EL & Euronext Paris & Europe \\
LVMH Mo\"et Hennessy Louis Vuitton SE & MC & Euronext Paris & Europe \\
TotalEnergies SE & TTE & Euronext Paris & Europe \\
Safran SA & SAF & Euronext Paris & Europe \\
Danone SA & BN & Euronext Paris & Europe \\
Sanofi SA & SAN & Euronext Paris & Europe \\
BNP Paribas SA & BNP & Euronext Paris & Europe \\
Orange SA & ORA & Euronext Paris & Europe \\
Renault SA & RNO & Euronext Paris & Europe \\
Engie SA & ENGI & Euronext Paris & Europe \\
STMicroelectronics NV & STM & Euronext Paris & Europe \\
Air Liquide SA & AI & Euronext Paris & Europe \\
Soci\'et\'e G\'en\'erale SA & GLE & Euronext Paris & Europe \\
Occidental Petroleum Corporation & OXY & NYSE & United States \\
Alibaba Group Holding Ltd. & BABA & NYSE & United States \\
Uber Technologies, Inc. & UBER & NYSE & United States \\
Exxon Mobil Corporation & XOM & NYSE & United States \\
NIKE, Inc. & NKE & NYSE & United States \\
Procter \& Gamble Company & PG & NYSE & United States \\
Chevron Corporation & CVX & NYSE & United States \\
Shopify Inc. & SHOP & NYSE & United States \\
The Coca-Cola Company & KO & NYSE & United States \\
ConocoPhillips & COP & NYSE & United States \\
Pfizer Inc. & PFE & NYSE & United States \\
Citigroup Inc. & C & NYSE & United States \\
Visa Inc. & V & NYSE & United States \\
Devon Energy Corporation & DVN & NYSE & United States \\
General Motors Company & GM & NYSE & United States \\
Synchrony Financial & SYF & NYSE & United States \\
Johnson \& Johnson & JNJ & NYSE & United States \\
Freeport-McMoRan Inc. & FCX & NYSE & United States \\
Carnival Corporation & CCL & NYSE & United States \\
Schlumberger Limited & SLB & NYSE & United States \\
\hline
\end{tabular}
\caption{List of stocks used in the estimation of the Hurst exponent of signed order flow.}
\label{tab:stocks}
\end{table}
 
\end{document}